\PassOptionsToPackage{capitalize,noabbrev}{cleveref}

\documentclass[]{preprint}

\usepackage[toc,page,header]{appendix}

\usepackage{microtype}
\usepackage{graphicx}
\usepackage{subcaption}
\usepackage{booktabs} 

\usepackage{hyperref}

\usepackage{amsmath}
\usepackage{amssymb}
\usepackage{mathtools}
\usepackage{amsthm}

\theoremstyle{plain}
\newtheorem{theorem}{Theorem}[section]
\newtheorem{lemma}[theorem]{Lemma}

\newtheorem{proposition}[theorem]{Proposition}

\theoremstyle{definition}
\newtheorem{definition}{Definition}[section]

\theoremstyle{remark}

\usepackage{hyperref}
\usepackage{url}
\usepackage{xcolor}
\usepackage{pifont}

\usepackage[utf8]{inputenc}
\usepackage{csquotes} 
\usepackage{enumitem}
\usepackage{multirow}
\usepackage{array}
\usepackage{colortbl}

\usepackage{arydshln} 

\usepackage{wrapfig}
\usepackage{subcaption}  
\usepackage{dblfloatfix}

\usepackage[most]{tcolorbox}
\usepackage{xcolor}
\usepackage{graphicx}

\definecolor{cvprblue}{RGB}{0,102,204} 

\usepackage{amsmath,bm}

\usepackage{color}
\usepackage{tikz}
\usetikzlibrary{shapes, arrows.meta, positioning, calc, fit, backgrounds}

\usepackage{etoc}

\definecolor{sybcblue}{HTML}{004488}  
\definecolor{brickred}{HTML}{BB5566}  
\definecolor{softgray}{HTML}{666666}  

\definecolor{graybg}{gray}{0.9}

\usepackage{hyperref}
\hypersetup{
    colorlinks=true,
    linkcolor=black,
    filecolor=magenta,      
    urlcolor=seedblue, 
    pdftitle={Modality Gap},
}

\usepackage{tikz}

\usepackage{cleveref}

\definecolor{softred}{RGB}{255,230,230}
\definecolor{softblue}{RGB}{230,240,255}
\definecolor{slowcolor}{RGB}{0,100,200} 

\usepackage{amsmath}   
\usepackage{amssymb}   
\usepackage{amsthm}    

\usepackage{tcolorbox}

\usepackage{enumitem}

\usepackage{amsmath,amssymb,amsthm,mathtools}

\usepackage[utf8]{inputenc} 
\usepackage[T1]{fontenc}    
\usepackage{hyperref}       
\usepackage{url}            
\usepackage{booktabs}       
\usepackage{amsfonts}       
\usepackage{nicefrac}       
\usepackage{microtype}      
\usepackage{xcolor}         

\usepackage[utf8]{inputenc} 
\usepackage[T1]{fontenc}    
\usepackage{hyperref}       
\usepackage{url}            
\usepackage{booktabs}       
\usepackage{amsfonts}       
\usepackage{nicefrac}       
\usepackage{microtype}      
\usepackage{xcolor}         

\usepackage{tikz}
\usetikzlibrary{arrows.meta, calc, shapes.geometric, decorations.pathreplacing, positioning}

\usepackage[utf8]{inputenc} 
\usepackage[T1]{fontenc}    
\usepackage{url}            
\usepackage{booktabs}       
\usepackage{amsfonts}       
\usepackage{nicefrac}       
\usepackage{microtype}      
\usepackage{xcolor}         

\usepackage{amsmath}
\usepackage{amssymb}

\usepackage[utf8]{inputenc} 
\usepackage[T1]{fontenc}    
\usepackage{url}            
\usepackage{booktabs}       
\usepackage{amsfonts}       
\usepackage{nicefrac}       
\usepackage{microtype}      

\usepackage{wrapfig}

\usepackage{pifont}

\usepackage[table]{xcolor}
\usepackage{booktabs}
\usepackage{multirow}
\usepackage{arydshln}

\usepackage{tikz}
\usetikzlibrary{spy, arrows.meta, calc, decorations.pathreplacing, bending, shadows} 

\usepackage{amsmath,amssymb}
\usepackage{tikz}
\usetikzlibrary{
  arrows.meta,
  positioning,
  calc,
  fit,
  backgrounds,
  decorations.pathreplacing,
  shapes.geometric,
  shapes.misc,
  matrix
}

\usepackage{tikz}

\usepackage{xcolor}

\usetikzlibrary{
    arrows.meta,      
    calc,             
    shapes.geometric, 
    fadings,          
    bending           
}

\usepackage{tikz}
\usetikzlibrary{arrows.meta, calc, decorations.pathreplacing, positioning}
\usepackage{amsmath, amssymb}

\usepackage{tikz}
\usetikzlibrary{arrows.meta, calc, angles, quotes}

\usepackage{booktabs}
\usepackage{multirow}
\usepackage{graphicx}
\usepackage[table]{xcolor}

\definecolor{graybg}{gray}{0.92}

\usepackage{tikz}
\usetikzlibrary{arrows.meta, calc}
\usepackage{amsmath, amssymb}

\definecolor{src}{HTML}{1B4965}
\definecolor{tgt}{HTML}{C44536}
\definecolor{gold}{HTML}{B8860B}
\definecolor{ell}{HTML}{C5D8E8}
\definecolor{gridg}{HTML}{DCDCDC}
\definecolor{note}{HTML}{6E7F8E}

\title{Anisotropic Modality Align}

\author[1,2]{Xiaomin Yu}
\author[3]{Yijiang Li}
\author[4]{Yuhui Zhang}
\author[2]{Hanzhen Zhao}
\author[4]{Yue Yang}
\author[5]{Hao Tang}

\makeatletter
\renewcommand\author[2][]{\addtolist[#1]{#2}{\authorlist}{\authorformat}{\\[0pt]}}
\makeatother

\author[6]{Yue Song}

\makeatletter
\renewcommand\author[2][]{\addtolist[#1]{#2}{\authorlist}{\authorformat}{, }}
\makeatother
\author[2]{Xiaobin Hu}
\author[1]{Chengwei Qin}
\author[2]{Shuicheng Yan}
\author[1]{Hui Xiong}

\affiliation[1]{HKUST(GZ)}
\affiliation[2]{NUS}
\affiliation[3]{UCSD}
\affiliation[4]{Stanford}
\affiliation[5]{PKU}
\affiliation[6]{THU}

\newcommand{\method}{AnisoAlign}

\abstract{
Training multimodal large language models has long been limited by the scarcity of high-quality paired multimodal data. Recent studies show that the shared representation space of pretrained multimodal contrastive models can serve as a bridge, enabling models to perform multimodal training with unimodal data. However, the key premise of this paradigm remains insufficiently understood: can representations from different modalities be reliably interchanged? The core obstacle lies in the persistent \textbf{Modality Gap} in the shared space. In this work, we revisit the geometric nature of the modality gap. We find that modality representations already share compatible dominant semantic geometry. What truly hinders modality interchangeability is not a simple global shift, but an anisotropic residual structure concentrated along a small number of dominant directions. Based on this finding, we further propose the principle of anisotropic modality gap alignment: effective modality alignment should align with the target-modality distribution while preserving the semantic structure of the source modality. Guided by this principle, we propose an anisotropic geometric correction framework, \textbf{\method{}}, for unpaired modality alignment. This framework leverages the internal geometric prior of the target modality and performs bounded correction on source-modality representations, thereby constructing substitute representations in the target modality. Experiments confirm its benefits in both geometric diagnostics and text-only MLLM training. Overall, this work recasts the modality gap from an empirical observation into a correctable, structured geometric phenomenon and provides a new representation alignment perspective for training multimodal models with unimodal data.
}

\date{\today}
\checkdata[Leader]{Xiaomin Yu (\email{yuxm02@gmail.com})}
\correspondence{Yue Song, Xiaobin Hu, Chengwei Qin}

\checkdata[Github]{\url{https://github.com/Yu-xm/Modality_Gap_Theory.git}}

\begin{document}
\maketitle

\newpage

\section{Introduction}

Multimodal contrastive learning models \cite{radford2021learning,zhai2023sigmoid,huang2026llm2clippowerfullanguagemodel} typically map samples from different modalities into the same normalized representation space, so that semantically corresponding images and texts are close to each other in this space. However, a persistent phenomenon is that, even after large-scale contrastive pretraining, image and text representations often still maintain systematic geometric separation in the shared space. This phenomenon is commonly referred to as the \textbf{Modality Gap} \cite{liang2022mind,zhang2024connect,yu2026modalitygapdrivensubspacealignment}. Some studies exploit this property by geometrically correcting the source-modality representations in the shared representation space and aligning them with the target modality, thereby enabling multimodal large language models (MLLMs) \cite{liu2023visual,he2024efficientmultimodallearningdatacentric} to be trained using single-modality data and decoupling the dependence on paired multimodal data \cite{chen2024sharegpt4v,he2024efficientmultimodallearningdatacentric}.

However, existing methods still lack a systematic characterization of the modality gap: do the two modalities share compatible dominant semantic geometry? Does the remaining discrepancy mainly arise from a global centroid shift, or is it concentrated as structured residuals along specific directions? What kind of correction can both preserve source-modality semantics and move representations into the distributional support of the target modality? Answering these questions is particularly critical for unpaired modality alignment, because in the absence of paired supervision, alignment methods must rely on the intrinsic geometric structure of modality distributions to constrain the correction process \cite{zhang2024connect,yu2026modalitygapdrivensubspacealignment}. This leads to the basic question studied in this work: \textbf{What kind of geometric discrepancy is the modality gap?}

To answer this question, we revisit the modality gap through a sequence of geometric diagnostics. The results show that image and text representations are not arbitrary, unrelated distributions in the shared space. Instead, the two modalities already possess compatible dominant semantic geometry: their covariance spectra exhibit similar long-tail decay, and their principal subspace overlap is significantly higher than the random baseline. This indicates that multimodal contrastive pretraining has already established a shared dominant geometric backbone between the two modalities.

However, the remaining modality gap cannot be simply explained by a global centroid bias. We find that, after globally shifting text representations to the image-modality centroid, most of the cross-modal discrepancy still remains. Further spectral analysis shows that the mean-corrected residual is not isotropic noise, but an anisotropic structure concentrated along a small number of dominant directions. In other words, \emph{the modality gap mainly appears as a low-effective-dimensional, direction-dependent residual, rather than an unstructured random offset}.

These diagnostics naturally lead to a modality alignment principle: effective modality alignment should not merely minimize global distributional discrepancy, but should satisfy two requirements simultaneously. First, it must preserve the semantic geometry already present in the source modality. Second, it must correct the dominant anisotropic residual directions that prevent the source modality from being compatible with the target-modality distribution. Matching only the target distribution may destroy semantic correspondence; preserving only source semantics may fail to enter the distributional support of the target modality. Therefore, \emph{modality alignment is essentially a structured geometric correction problem between semantic preservation and target-distribution compatibility}.

Based on this principle, we propose an anisotropic alignment method, \textbf{\method{}}, for unpaired modality alignment. The method first constructs a fixed dominant subspace decomposition, dividing the shared space into a statistically dominant subspace and its orthogonal complement. Then, within the dominant subspace, we introduce a blockwise polar parameterization that decomposes representations into radius and phase structures, thereby explicitly modeling anisotropic geometric variations along dominant directions. To avoid directly learning an unstable cross-modal mapping, we first pretrain a periodic phase prior using only target-modality samples, which captures the internal phase statistics of the target modality. Then, in the second stage, we perform bounded residual correction on source-modality representations, so that they gradually satisfy the target-modality prior while preserving instance-level semantic structure.

Extensive experiments support this view. At the representation level, \method{} better matches the target-modality geometry while preserving source-modality semantics, achieving balanced local support compatibility and reducing dominant anisotropic residual directions. At the MLLM level, the resulting substitute representations lead to stronger performance in both fully text-only training and text-only pretraining before visual instruction tuning. These results suggest that modality alignment is better understood as structured anisotropic geometric correction, and that large-scale text-only data can be leveraged as a useful substitute for paired image-text supervision. 

\section{Preliminaries}

\begin{definition}[\textbf{Modality Gap}.] Let $X_0$ and $Y_0$ denote two distinct modalities, let $f_x:X_0\to S^{d-1}$ and $f_y:Y_0\to S^{d-1}$ be pretrained encoders into a shared normalized representation space, and write $X=f_x(X_0)$ and $Y=f_y(Y_0)$. Let $\sigma:S^{d-1}\to S$ denote the latent semantic map, where $S$ is an abstract semantic space and $\sigma(z)$ denotes the semantic label associated with $z$. If, for semantically corresponding cross-modal representations $x\in X$ and $y\in Y$, $\sigma(x)=\sigma(y)$ while $x$ and $y$ need not coincide geometrically, and this discrepancy is systematic at the distribution level, $\mu_x\neq\mu_y$ or $\Sigma_x\neq\Sigma_y$, where $\mu$ and $\Sigma$ denote the mean and covariance, respectively, then such a systematic cross-modal geometric discrepancy is called the \textbf{Modality Gap} phenomenon. \end{definition}

\begin{definition}[\textbf{Modality Align}.] In a shared representation space exhibiting modality gap, let $Y$ be the source modality and $X$ the target modality. \textbf{Modality Align} seeks a mapping $T:\mathbb R^d\to\mathbb R^d$ that rectifies the cross-modal geometric discrepancy such that, given only unpaired samples from $X$ and $Y$, for any $y\in Y$, $\sigma(T(y))=\sigma(y)$ and $P_{T(Y)\mid \sigma}\approx P_{X\mid \sigma}$. The transformed representation $T(y)$ is called a substitute representation of $y$ in the target modality. \end{definition}





\section{Modality Gap} \label{sec:3}

Two modalities in the shared embedding space often remain separated by a persistent modality gap. This raises a basic geometric question: \textbf{What kind of discrepancy is the modality gap?}



\subsection{Geometric Compatibility Across Modalities} \label{subsec:3.1}


\begin{wrapfigure}{r}{0.6\textwidth}
    \centering
    \includegraphics[width=\linewidth]{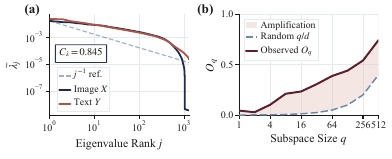}
    \caption{
        \textbf{Image and text modalities share compatible dominant geometry.}
        (a) The normalized covariance spectra of the two modalities exhibit similar long-tail decay, with spectral correlation $C_\lambda=0.845$.
        (b) Their principal subspace overlap is consistently above the random baseline $q/d$; at $q=128$, $O_{128}=0.441$ versus $q/d=0.100$, indicating shared non-random dominant directions.
        }
    \label{fig1}
\end{wrapfigure}

We first ask whether the two modalities have compatible global geometry in the shared representation space. This question is essential: if two embeddings were merely two arbitrary and unrelated distributions, then any geometric correction would not preserve semantic consistency. To test this, we compare the dominant covariance structure of the two modalities.

\textbf{Compatible Spectral Decay.}
Given $1\mathrm{M}$ paired image-text representations $\{(x_i,y_i)\}_{i=1}^{n}$, where $x_i\in X$, $y_i\in Y$. Let $\Sigma_x$ and $\Sigma_y$ denote the centered covariance matrices of the image and text modalities. We compare their covariance spectra by sorting the eigenvalues in descending order and defining the spectral correlation as $C_\lambda = \operatorname{corr}\left(\log\lambda(\Sigma_x),\log\lambda(\Sigma_y)\right)$. As shown in Fig.~\ref{fig1}(a), the normalized spectra of the two modalities exhibit similar long-tail decay. The spectral correlation reaches $C_\lambda=0.845$, indicating that image and text representations distribute their variance energy across dominant directions compatibly.

\textbf{Shared Principal Structure.}
Spectral similarity alone does not guarantee that the two modalities use the same directions. We therefore next ask whether their principal subspaces overlap. Let $U_x^{q}$ and $U_y^{q}$ denote the subspaces spanned by the top $q$ eigenvectors of $\Sigma_x$ and $\Sigma_y$, respectively. We define the subspace overlap as $O_q =
\frac{1}{q}
\left\|
(U_x^{q})^\top U_y^{q}
\right\|_F^2$. If the two subspaces were randomly unrelated, the expected overlap would be approximately $q/d$. However, Fig.~\ref{fig1}(b) shows that the observed $O_q$ is consistently above this random baseline across different subspace sizes. In particular, when $q=128$, we obtain $O_{128}=0.441$, whereas the random baseline is only $q/d=0.100$. Thus, image and text representations share a set of non-random dominant geometric directions.


\begin{tcolorbox}[boxrule=1pt, colframe=black!70, colback=gray!5, left=8pt, right=8pt, top=4pt, bottom=4pt]
\noindent\textbf{Conclusion 1. (Compatible Dominant Geometry).}
The modality gap does not mean that image and text representations have unrelated global geometry. Instead, the two modalities already share compatible dominant semantic structure in the shared representation space.
\end{tcolorbox}

\begin{figure}[htbp]
    \centering
    \includegraphics[width=\textwidth]{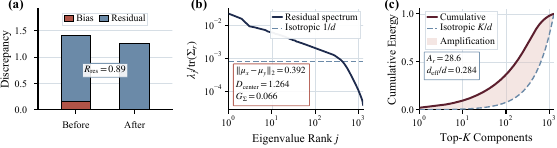}
    \caption{
        \textbf{The modality gap is dominated by an anisotropic residual.}
        (a) Mean correction removes only a small fraction of the cross-modal discrepancy, leaving a large residual gap.
        (b) The residual covariance spectrum deviates strongly from the isotropic baseline, with dominant eigen-directions.
        (c) Residual energy is concentrated in a low-effective-dimensional subspace, with anisotropy ratio $A_r=28.6$ and $d_{\mathrm{eff}}/d=0.284$.
        }
    \label{fig2}
    \vspace{-10pt}
\end{figure}

\subsection{Anisotropic Modality Gap}  \label{subsec:3.2}


Having established that the two modalities share compatible dominant geometry, we next ask what form the remaining modality gap takes. A natural hypothesis is that the gap is mainly a global centroid bias. Let $(\mu_x,\Sigma_x)$ and $(\mu_y,\Sigma_y)$ denote the empirical means and centered covariances of the two modalities, respectively. We measure centroid displacement and covariance-shape discrepancy as $G_\mu=\|\mu_x-\mu_y\|_2$ and
$G_\Sigma=\|\Sigma_x-\Sigma_y\|_F/(\|\Sigma_x\|_F+\epsilon)$.

\textbf{Centroid Bias Is Insufficient.}
If the modality gap were dominated by a global mean shift, then translating text representations to the image centroid should remove most of the cross-modal discrepancy. To test this hypothesis, we keep image representations fixed and apply mean correction to text representations as $y_i^{x}=y_i-\mu_y+\mu_x$. The paired residual after mean correction is $r_i=x_i-y_i^{x}=(x_i-\mu_x)-(y_i-\mu_y)$, with residual covariance $\Sigma_r=\frac{1}{n}\sum_{i=1}^{n}r_i r_i^\top$. Fig.~\ref{fig2}(a) confirms that the two modalities have a clear centroid displacement, with $G_\mu=0.392$. However, the covariance-shape discrepancy is also nonzero, with $G_\Sigma=0.066$, suggesting that the misalignment is not purely a difference in mean centers. Although text representations are globally shifted to the image centroid, the corrected paired distance remains high, $\widetilde{D}=1.264$. The residual ratio is $\widetilde{D}/D=0.89$. This rules out the simplest explanation that the modality gap is mainly a centroid bias.

\textbf{Anisotropic Residual.}
We next ask whether the remaining residual is isotropic noise. If this were the case, then its covariance would satisfy $\Sigma_r \approx \sigma^2 I$, and its normalized eigenvalue spectrum would be close to the flat isotropic baseline $1/d$. However, Fig.~\ref{fig2}(b) shows a different pattern. The residual spectrum has dominant eigen-directions whose energy is far above the isotropic average, followed by a long-tail decay. To quantify this deviation, we define the residual anisotropy ratio as \(A_r=\lambda_{\max}(\Sigma_r)/(\operatorname{tr}(\Sigma_r)/d)\), where $\lambda_{\max}(\Sigma_r)$ is the largest eigenvalue of the residual covariance. Fig.~\ref{fig2}(c) shows $A_r=28.6\gg 1$. Therefore, the residual gap is not random isotropic noise; it is strongly direction-dependent. This anisotropy is further reflected in residual energy concentration. We compute the cumulative energy explained by the top-$K$ residual eigen-directions, $E(K)=\sum_{j=1}^{K}\lambda_j(\Sigma_r)/\sum_{j=1}^{d}\lambda_j(\Sigma_r)$. As shown in Fig.~\ref{fig2}(c), the empirical curve lies far above the isotropic baseline $K/d$, indicating that residual energy is concentrated in a small number of dominant directions. We further compute the effective dimension 
$d_{\mathrm{eff}}(\Sigma_r)=\mathrm{tr}(\Sigma_r)^2/\mathrm{tr}(\Sigma_r^2)$, obtaining $d_{\mathrm{eff}}/d=0.284$, which confirms that the residual gap lies in a low-effective-dimensional anisotropic subspace.

\begin{tcolorbox}[boxrule=1pt, colframe=black!70, colback=gray!5, left=8pt, right=8pt, top=4pt, bottom=4pt]
\noindent\textbf{Conclusion 2 (Anisotropic Residual Gap).}
The modality gap is dominated by a structured residual: a direction-dependent anisotropic discrepancy concentrated in a low-effective-dimensional subspace.
\end{tcolorbox}

\subsection{Anisotropic Modality Alignment Principle}

The previous diagnostics reveal two facts. First, image and text representations already share compatible dominant semantic geometry. Second, the remaining modality gap is a low-effective-dimensional anisotropic residual. We therefore ask: \textbf{What should effective modality alignment preserve, and what should it correct?}

\begin{figure}[htbp]
    \centering
    \includegraphics[width=\textwidth]{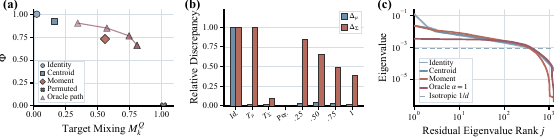}
    \caption{
        \textbf{Effective alignment requires both source semantic preservation and target distribution compatibility.}
        (a) Different transformations exhibit a trade-off between source instance consistency and target local mixing.
        (b) Centroid and moment corrections reduce global discrepancies, while random target replacement destroys semantic correspondence.
        (c) Correction along the anisotropic residual subspace reduces dominant residual directions while better preserving source-side semantics.
        }
    \label{fig3}
    \vspace{-10pt}
\end{figure}

To answer this question, we compare five diagnostic transformations: 
\ding{182} Identity Mapping \(T_{\mathrm{id}}\): the unaligned state; 
\ding{183} Centroid Correction \(T_\mu\): only removes the global centroid shift; 
\ding{184} Moment Correction \(T_\Sigma\): matches global moment statistics; 
\ding{185} Random Target Replacement \(T_{\mathrm{perm}}\): serves as a negative control that matches the target distribution but destroys semantic correspondence; and 
\ding{186} \(T_\alpha\): provides a controlled interpolation between semantic preservation and target-distribution compatibility by correcting representations along dominant residual directions. The experimental results show that different transformations exhibit clearly different alignment behaviors. As shown in Fig.~\ref{fig3}(a), \(T_\mu\) preserves source-side semantics well, but provides limited improvement in target-side local mixing; \(T_{\mathrm{perm}}\), although drawn from the target distribution, almost completely destroys source semantics, indicating that matching the target distribution alone is insufficient. Fig.~\ref{fig3}(b) further shows that \(T_\Sigma\) reduces global statistical discrepancy, but introduces noticeable source-side semantic degradation. In contrast, \(T_\alpha\) forms a continuous trade-off between source-side semantic preservation and target-side geometric compatibility. Finally, Fig.~\ref{fig3}(c) shows that correcting along dominant anisotropic residual directions more directly suppresses the dominant residual components. Therefore, effective alignment should not be viewed as minimizing a single global gap; instead, it should both preserve the semantic geometry of the source modality and correct the dominant anisotropic residuals that prevent compatibility with the target distribution. We provide theoretical support for the geometric diagnostics and the anisotropic alignment principle in Appendix.~\ref{app:1}. The above diagnostics naturally lead to the following principle:

\begin{tcolorbox}[boxrule=1pt, colframe=black!70, colback=gray!5, left=8pt, right=8pt, top=4pt, bottom=4pt]
\noindent\textbf{Principle (Anisotropic Modality Alignment).}
Effective modality alignment should preserve the source modality's semantic geometry while correcting the dominant anisotropic residual directions that prevent compatibility with the target-modality distribution.
\end{tcolorbox}


\section{\method{}}

\subsection{Fixed-Frame Subspace Decomposition}


Following Sec.~\ref{subsec:3.1}, we first fix a shared dominant subspace to provide a stable geometric frame for alignment, and identify a shared dominant subspace capturing the major geometric structure of both modalities. Let $\mu_t, \mu_i \in \mathbb{R}^d$ denote the empirical means of text embeddings and image embeddings, respectively, and let $\Sigma_t, \Sigma_i \in \mathbb{R}^{d\times d}$ denote the corresponding centered covariance matrices. We define the joint structure matrix as $\Sigma = \Sigma_t + \Sigma_i + \lambda I$, where $\lambda > 0$ is a regularization parameter and $I$ is the identity matrix. Let $Q_U \in \mathbb{R}^{d\times r}$ consist of the top-$r$ eigenvectors of $\Sigma$. Then, $\mathbb{R}^d$ can be decomposed into two mutually orthogonal subspaces: $\mathbb{R}^d = U \oplus V$, with $U = \mathrm{span}(Q_U)$. Under this decomposition, any embedding $z \in \mathbb{R}^d$ can be uniquely written as:
\begin{equation}
z_U = Q_U Q_U^\top z,\qquad z_V = z - z_U.
\end{equation}
Here, $z_U$ denotes the orthogonal projection of $z$ onto the subspace $U$, capturing its component along the first $r$ dominant statistical directions; $z_V$ denotes the remaining component orthogonal to $U$. All subsequent alignment operations are performed under this fixed decomposition.


\subsection{Anisotropic Circular Decoupling}

\begin{wrapfigure}{r}{0.4\textwidth}
  \centering
  \vspace{-40pt}
  \resizebox{\linewidth}{!}{%
    \begin{tikzpicture}[>=Stealth, scale=1.6]
        
        \colorlet{colDist}{teal!85!black}     
        \colorlet{colVec}{blue!75!black}      
        \colorlet{colProj}{orange!90!black}   
        \colorlet{colAxis}{gray!60}           
        \colorlet{colMix}{magenta!80!black}   
        
        \def\rot{22} 
        
        \begin{scope}[rotate=\rot]
            \draw[colDist!40, fill=colDist!4, line width=1pt] (0,0) ellipse (2.2cm and 1.1cm);
            \draw[colDist!70, fill=colDist!12, line width=1pt] (0,0) ellipse (1.4cm and 0.7cm);
            
            \draw[dashed, colDist!80, line width=0.8pt] (-0.4, 0) -- (2.4, 0) node[above right, font=\scriptsize, inner sep=1pt] {$Q_U$};
            \draw[dashed, colDist!80, line width=0.8pt] (0, -0.4) -- (0, 1.3);
        \end{scope}
        
        \node[colDist, font=\scriptsize, anchor=north west] at (1.1, -0.3) {Anisotropic Subspace $U$};
        
        \draw[->, colAxis, line width=1pt] (-0.5, 0) -- (2.7, 0) node[right, text=black] {$a_k$};
        \draw[->, colAxis, line width=1pt] (0, -0.5) -- (0, 1.9) node[above, text=black] {$b_k$};
        \filldraw[gray!50] (0,0) circle (1pt); 
        
        \draw[->, colMix, line width=1pt] (\rot:2.25) arc (\rot:0:2.25);
        \node[colMix, font=\scriptsize, right, inner sep=2pt] at ({\rot/2}:2.27) {$Q_U R$};
        
        \def\r{1.9}
        \def\ang{46}
        \coordinate (P) at (\ang:\r);
        
        \draw[dashed, colProj, line width=0.8pt] (P) -- ({\r*cos(\ang)}, 0) node[below] {$a_k$};
        \draw[dashed, colProj, line width=0.8pt] (P) -- (0, {\r*sin(\ang)}) node[left] {$b_k$};

        \draw[dashed, colVec!50, line width=1pt] (0:\r) arc (0:82:\r);
        \node[colVec!80, font=\scriptsize, rotate=-20] at (76:\r+0.15) {Constant $\rho_k$};
        
        \draw[->, colVec, line width=1.2pt] (0,0) -- (P);
        \filldraw[colVec] (P) circle (1.5pt) node[above right, font=\small, inner sep=2pt] {$c_k = (a_k, b_k)$};
        
        \path (0,0) -- (P) node[midway, sloped, above, colVec, font=\small, inner sep=2pt] {$\rho_k$};
        \draw[->, colVec, line width=1pt] (0:0.6) arc (0:\ang:0.6);
        \node[colVec, font=\small] at ({\ang/2}:0.8) {$\theta_k$};
        
    \end{tikzpicture}%
  }
  \vspace{-10pt}
  \caption{Anisotropic circular decoupling in $U$ subspace.}
  \label{fig:ecoupling}
  \vspace{-10pt}
\end{wrapfigure}
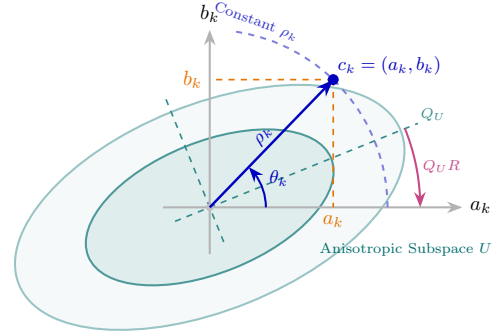


Following Sec.~\ref{subsec:3.2}, we then use blockwise polar coordinates to explicitly model the anisotropic residual structure. We introduce an explicit blockwise polar parameterization protocol within the dominant subspace $U$. As shown in Fig.~\ref{fig:ecoupling}. We first map the projection $Q_U^\top z \in \mathbb{R}^r$ into $m = r/2$ discrete two-dimensional subspaces. However, natively constructing these subspaces directly based on the principal component hierarchy introduces an arbitrary dependence on specific eigenvector orderings, making the decomposition sensitive to arbitrary eigenvector orderings. To inoculate the architecture against this basis dependence, we introduce a continuous orthogonal mixing matrix $R \in \mathbb{R}^{r \times r}$, subject to the strict constraint $R^\top R = I$. We dynamically redefine the internal coordinate basis as $Q_U \leftarrow Q_U R$. This mixing operation preserves the invariant span of the subspace $U$ while autonomously discovering a maximally stable internal coordinate organization for downstream anisotropic decoupling. Based on this optimized coordinate system, let $(a_k, b_k)$ denote the coordinates of the projected vector $c = Q_U^\top z \in \mathbb{R}^r$ within the $k$-th two-dimensional block. We reformulate these Euclidean coordinates into a polar embedding:
\begin{equation}
    \rho_k = \sqrt{a_k^2 + b_k^2 + \varepsilon}, \quad \theta_k = \text{atan2}(b_k, a_k)
\end{equation}
where $\varepsilon>0$ ensures numerical stability near the origin. The embedding in $U$ is thus decoupled into blockwise radii $\rho=(\rho_1,\dots,\rho_m)$ and phases $\theta=(\theta_1,\dots,\theta_m)$.

\subsection{Stage I: Target-Modality Periodic Prior Pretraining}

 \begin{wrapfigure}{r}{0.4\textwidth}
  \centering
  \vspace{-10pt}
  \resizebox{\linewidth}{!}{%
    \begin{tikzpicture}[>=Stealth, scale=1.8] 
        
        \colorlet{colTarget}{teal!90!black}     
        \colorlet{colPhase}{blue!75!black}      
        \colorlet{colMu}{orange!90!black}        
        \colorlet{colDist}{red!75!orange}       
        \colorlet{colNeighbor}{purple!75!black} 
        \colorlet{colArc}{gray!40}              

        \draw[dashed, gray!30, thin] (0,0) -- (0:1.5);
        \draw[dashed, gray!30, thin] (0,0) -- (90:1.5);
        \draw[dashed, gray!30, thin] (0,0) -- (45:1.5);
        \draw[dashed, gray!30, thin] (0,0) -- (135:1.5); 
        \filldraw[gray!50] (0,0) circle (0.8pt); 

        \draw[line width=1.2pt, colArc] (-15:1.5) arc (-15:150:1.5);
        
        \draw[line width=1.2pt, colDist, fill=colDist!15, fill opacity=0.8, domain=15:75, samples=80] 
            plot ({1.5*cos(\x)*(1 + 0.25*exp(-((\x-45)*(\x-45))/50))}, 
                  {1.5*sin(\x)*(1 + 0.25*exp(-((\x-45)*(\x-45))/50))}) 
            -- (75:1.5) arc (75:15:1.5) -- cycle; 

        \filldraw[colTarget] (0:1.5) circle (1.5pt);
        \node[colTarget, font=\small, right=3pt] at (0:1.5) {$\bar{\psi}_k, \alpha_k$};

        \filldraw[colPhase] (90:1.5) circle (1.5pt);
        \node[colPhase, font=\small, right=2pt, yshift=-2pt] at (90:1.5) {$\phi_k$}; 

        \filldraw[colNeighbor] (135:1.5) circle (1.5pt);
        \node[colNeighbor, font=\small, left=2pt, yshift=-2pt] at (135:1.5) {$\phi_\ell$};

        \draw[<->, dashed, colNeighbor, line width=0.8pt] (132:1.62) arc (132:93:1.62);
        \node[colNeighbor, font=\footnotesize, above left=1pt] at (112.5:1.62) {$A_{k\ell}, \eta_{k\ell}$};

        \filldraw[colMu] (45:1.5) circle (1.5pt);
        \node[colMu, font=\small, below left=1pt] at (45:1.45) {$\mu_\phi$};

        \draw[->, line width=1pt, colPhase] (85:1.68) arc (85:55:1.68);
        \node[colPhase, font=\footnotesize, inner sep=1pt] at (70:1.85) {$-\tau \nabla_\phi \Psi$};

        \node[colDist, font=\footnotesize, anchor=west] at (1.4, 1.) {$\tilde{\phi}, \sigma_t$};

    \end{tikzpicture}%
  }

  \caption{Target-modality periodic prior in phase space. Image phases define marginal anchors \((\bar{\psi}_k,\alpha_k)\) and pairwise couplings \((A_{k\ell},\eta_{k\ell})\), which induce a drift field \(-\tau\nabla_\phi\Psi\) and train the frozen phase score prior \(s_\phi\).}
  \label{fig:periodic_prior}
  \vspace{-8pt}
\end{wrapfigure}
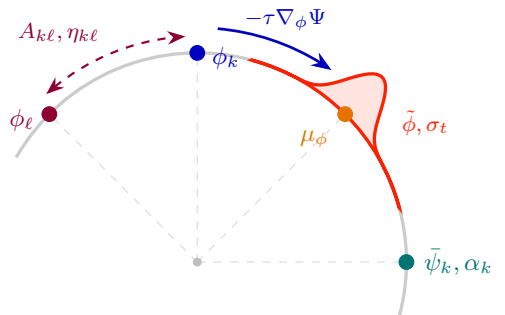

Before learning any modality alignment, we first estimate the phase statistical structure of the target modality in the decoupled phase space using only the image. As shown in Fig.~\ref{fig:periodic_prior}. This structure consists of two aspects: first, the marginal distributions of the phase variables of individual two-dimensional blocks; second, the dependency relations among phase differences across different two-dimensional blocks. Stage I does not involve learning a text-to-image mapping. Instead, it constructs a frozen periodic score prior $s_\phi$ from the image modality, which is subsequently used in Stage II as a target-modal constraint.

For an image embedding $x$, let $\{(\rho_k^{(x)}, \theta_k^{(x)})\}_{k=1}^m$ denote its polar embedding. We define the blockwise circular correlation statistic as:
\begin{equation}
M_{k\ell}^{(x)}=\mathbb{E}\!\left[e^{\,i(\theta_k^{(x)}-\theta_\ell^{(x)})}\right]\in\mathbb{C}. 
\end{equation}
Here, $|M_{k\ell}^{(x)}|$ measures the consistency of the phase difference between the $k$-th and $\ell$-th blocks, while $\arg(M_{k\ell}^{(x)})$ gives the corresponding empirical phase offset. Instead of selecting globally top-$p$ block pairs over all possible pairs, we construct the sparse dependency graph in a block-adaptive manner: for each block $k$, we retain the top-$p$ blocks $\ell\neq k$ with the largest $|M_{k\ell}^{(x)}|$, and then take the union of all retained undirected pairs. This yields a sparse dependency graph $E\subseteq [m]\times[m]$, where $[m]:=\{1,\ldots,m\}$.

Based on these quantities, we define a drift field in phase space, $\nabla_\phi\Psi(\phi)\in\mathbb{R}^m$, where $\phi=(\phi_1,\ldots,\phi_m)\in[-\pi,\pi)^m$. Its $k$-th component is
\begin{equation}
[\nabla_\phi \Psi(\phi)]_k
=
\alpha_k \sin(\phi_k-\bar{\psi}_k)
+
\sum_{\ell:(k,\ell)\in E}
A_{k\ell}\sin(\phi_k-\phi_\ell-\eta_{k\ell}). 
\end{equation}
Here, $A_{k\ell}=|M_{k\ell}^{(x)}|\in\mathbb{R}_{\ge 0}$ and $\eta_{k\ell}=\arg(M_{k\ell}^{(x)})\in[-\pi,\pi)$ denote the coupling strength and empirical phase offset of edge $(k,\ell)$, respectively; $\bar{\psi}_k=\arg(\mathbb{E}[e^{\,i\theta_k^{(x)}}])\in[-\pi,\pi)$ denotes the dominant phase location of the $k$-th two-dimensional block; and $\alpha_k=\mathbb{E}[(\rho_k^{(x)})^2]/(\sum_{u=1}^m \mathbb{E}[(\rho_u^{(x)})^2]+\varepsilon)\in\mathbb{R}_{\ge 0}$ denotes the relative weight of that block.

Given a phase vector $\phi$, we first define the drifted phase center $\mu_\phi\in[-\pi,\pi)^m$ as:
\begin{equation}
\mu_\phi=\operatorname{wrap}\!\big(\phi-\tau\nabla_\phi\Psi(\phi)\big). 
\end{equation}
We then construct a perturbed phase sample $\tilde{\phi}\in[-\pi,\pi)^m$ as $\tilde{\phi}=\operatorname{wrap}(\mu_\phi+\sqrt{2}\sigma_t\epsilon)$, where $\epsilon\sim\mathcal{N}(0,I_m)$, $\tau>0$ is the drift step size, and $\sigma_t>0$ is the noise scale at time step $t$.

On this basis, we train a phase-aware score network $s_\phi:\mathbb{R}^m\times\mathbb{R}\times\mathbb{R}^m\to\mathbb{R}^m$, whose input is $(\tilde{\phi},t,\log\rho)$ and whose output is the phase score $s_\phi(\tilde{\phi},t,\log\rho)\in\mathbb{R}^m$. The Stage-I loss is defined as:
\begin{equation}
\mathcal{L}^{\mathrm{I}}
=
\mathbb{E}_{t,\tilde{\phi}}
\left[
\lambda_t
\left\|
s_\phi(\tilde{\phi},t,\log\rho)
-
\nabla_{\tilde{\phi}}\log q(\tilde{\phi}\mid \mu_\phi,\sigma_t)
\right\|_2^2
\right], 
\end{equation}
where $q(\tilde{\phi}\mid \mu_\phi,\sigma_t)$ denotes a wrapped Gaussian distribution centered at $\mu_\phi$ with noise scale $\sigma_t$, $\nabla_{\tilde{\phi}}\log q(\tilde{\phi}\mid \mu_\phi,\sigma_t)\in\mathbb{R}^m$ is its score with respect to $\tilde{\phi}$, and $\lambda_t=2\sigma_t^2$.


Therefore, Stage I yields a phase score prior determined by the target image distribution. This prior is kept frozen after training and is introduced in Stage II as a target-modal constraint.

\subsection{Stage II: Prior-Guided Bounded Alignment}

After fixing the periodic prior $s_\phi$ of the target modality, Stage II performs a two-stage update on the text embedding $y\in\mathbb{R}^d$: a deterministic global initialization followed by an instance-conditioned bounded refinement.

\subsubsection{Global Initialization}

We first recenter the text embedding by $\bar y = y-\mu_t+\mu_i \in \mathbb{R}^d$. \textbf{On $U$-side}. We project $\bar y$ onto the mixed basis and express it in blockwise polar coordinates $(\rho,\theta)$. We set $\theta^{(0)}=\theta\in[-\pi,\pi)^m$ and define $\rho_k^{(0)}=T_k(\rho_k)$, where $T_k(r)=\big(F_k^{(x)}\big)^{-1}\!\big(F_k^{(y)}(r)\big)$. Here, $F_k^{(x)}$ and $F_k^{(y)}$ denote the empirical radial cumulative distribution functions of images and text, respectively, on the $k$-th two-dimensional block. This gives $\rho^{(0)}=(\rho_1^{(0)},\ldots,\rho_m^{(0)})\in\mathbb{R}_{>0}^m$. \textbf{On $V$-side}. We define $y_U=Q_UQ_U^\top y$ and $y_V=y-y_U$, and set $v^{(0)}=\mu_{i,V}+D_V\bigl(y_V-\mu_{t,V}\bigr)\in\mathbb{R}^d$, where $D_V=\mathrm{Diag}\!\bigl(\sigma_V^{(x)}/(\sigma_V^{(y)}+\varepsilon)\bigr)$, $\mu_{i,V}=P_V\mu_i$, and $\mu_{t,V}=P_V\mu_t$. This yields the initialized state $(\theta^{(0)},\rho^{(0)},v^{(0)})$.

\subsubsection{Prior-Guided Residual Refinement}

Starting from the initialized state, we use an instance-conditioned map $g_\eta$ to predict residual corrections for phase, radius, and the $V$-subspace component:
\begin{equation}
(\Delta\theta,\Delta\rho,\Delta v)
=
g_\eta\!\big([\sin\theta^{(0)};\cos\theta^{(0)};\log\rho^{(0)};v^{(0)}]\big),
\end{equation}
where $\Delta\theta,\Delta\rho\in\mathbb{R}^m$ and $\Delta v\in\mathbb{R}^d$. Since the refinement of the residual component is restricted to the orthogonal complement $V$, we remove its $U$-projection and keep only the $V$-part, i.e., $\Delta v_V=\Delta v-Q_UQ_U^\top\Delta v$. Rather than directly denoising toward the target modality, we constrain the refined phase configuration to remain locally compatible with the target prior. The refined phase, radius, and residual component are then given by $\hat{\theta}=\operatorname{wrap}\!\big(\theta^{(0)}+\alpha_\theta\tanh(\Delta\theta)\big)$, $\hat{\rho}_k=\rho_k^{(0)}\exp\!\big(\alpha_\rho\tanh(\Delta\rho_k)\big)$, and $\hat v=v^{(0)}+\alpha_v\tanh(\Delta v_V)$. 
so that $\hat{\theta}\in[-\pi,\pi)^m$, $\hat{\rho}=(\hat{\rho}_1,\ldots,\hat{\rho}_m)\in\mathbb{R}_{>0}^m$, and $\hat v\in\mathbb{R}^d$.


To impose the target-modality prior, instead of using a one-step denoising guidance objective, we construct a prior-matching loss around the refined phase itself. Specifically, we first define $\mu_{\hat\theta}=\operatorname{wrap}\!\big(\hat\theta-\tau\nabla_\phi\Psi(\hat\theta)\big)$ and then perturb it as $\tilde{\theta}=\operatorname{wrap}\!\big(\mu_{\hat\theta}+\sqrt{2}\sigma_t\epsilon\big)$, where $\epsilon\sim\mathcal{N}(0,I_m)$. We define the prior-matching loss as
\begin{equation}
\mathcal{L}^{\mathrm{II}}
=
\mathbb{E}_{t,\epsilon}
\left[
\lambda_t
\left\|
s_\phi(\tilde{\theta},t,\log\hat{\rho})
-
\nabla_{\tilde{\theta}}\log q(\tilde{\theta}\mid \mu_{\hat\theta},\sigma_t)
\right\|_2^2
\right].
\end{equation}
This objective encourages the refined phase configuration to remain locally compatible with the frozen target-modality periodic prior.

In parallel, reusing the sparse graph $E$ from Stage I, we define $\omega_{k\ell}^{(0)}=\rho_k^{(0)}\rho_\ell^{(0)}\big/\bigl(\sum_{(u,v)\in E}\rho_u^{(0)}\rho_v^{(0)}+\varepsilon\bigr)$ for any $(k,\ell)\in E$, and impose the relative phase deformation constraint
\begin{equation}
\mathcal{L}^{\Phi}=\frac{1}{|E|}\sum_{(k,\ell)\in E}\omega_{k\ell}^{(0)}\left[1-\cos\!\Big((\hat{\theta}_k-\hat{\theta}_\ell)-(\theta_k^{(0)}-\theta_\ell^{(0)})\Big)
\right].
\end{equation}
Finally, let $c(\hat{\rho},\hat{\theta})\in\mathbb{R}^r$ denote the blockwise Cartesian vector generated from $(\hat{\rho},\hat{\theta})$, whose $k$-th two-dimensional block is $(\hat{\rho}_k\cos\hat{\theta}_k,\hat{\rho}_k\sin\hat{\theta}_k)$. We first reconstruct an intermediate normalized embedding as $e'=\operatorname{Norm}\!\big(Q_U\,c(\hat{\rho},\hat{\theta})+\hat v\big)\in\mathbb{S}^{d-1}$. After transforming the full text corpus, we further estimate the global mean of the intermediate transformed representations, $\hat\mu=\mathbb{E}_y[e'(y)]$, and perform a final global centroid calibration by defining $e=\operatorname{Norm}\!\big(e'-\hat\mu+\mu_i\big)\in\mathbb{S}^{d-1}$. The calibrated representation $e$ is used as the final substitute representation in the target modality.

\section{Experiments}

In this section, we systematically evaluate the effectiveness of our method from two perspectives: representation-level geometric diagnostics and MLLM training. The experiments are designed to answer six core questions, Q1--Q6. At the representation level, we use images as the target modality and texts as the source modality. We randomly sample 10K paired image-text representation samples for geometric diagnostics. At the MLLM level, we keep the model architecture, decoding settings, training data, and evaluation protocol unchanged. We use LLM2CLIP-Openai-L-14-336 \cite{huang2026llm2clippowerfullanguagemodel} as the encoder and Llama-3-8B-Instruct as the LLM backbone. We compare four methods: Text, C3 \cite{zhang2024connect}, ReAlign \cite{yu2026modalitygapdrivensubspacealignment}, and \method{}. Detailed experiment settings are provided in Appendix.~\ref{app:2}. 


\paragraph{\ding{182} \textbf{\method{} Better Match the Target-Modality Geometry?}}

\begin{figure}[t]
    \centering
    \includegraphics[width=1.\linewidth]{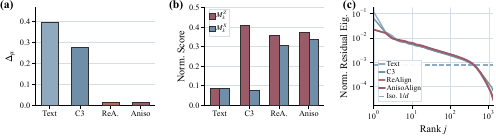}
    \vspace{-15pt}
    \caption{
        Target-geometry compatibility of different alignment methods. AnisoAlign achieves near-zero centroid discrepancy, the most balanced local support matching, and a low-anisotropy residual spectrum, outperforming Text and C3 while remaining competitive with ReAlign.
        }
    \vspace{-10pt}
    \label{fig6}
\end{figure}

This experiment examines whether the transformed source representations \(Z=T(Y)\) enter the geometric support of the target modality. We first measure the centroid discrepancy \(\Delta_\mu(T)=\|\mu_z-\mu_x\|_2\). As shown in Fig.~\ref{fig6}(a), Text shows a global offset with \(\Delta_\mu=0.393\), and C3 reduces it to \(0.276\). In contrast, ReAlign and \method{} both reduce it to about \(0.012\), indicating effective target-centroid calibration. We evaluate local support compatibility. As shown in Fig.~\ref{fig6}(b), C3 obtains \(M_k^Z=0.410\) but only \(M_k^X=0.075\), suggesting sparse target penetration without sufficient target coverage. ReAlign gives more balanced scores, \(M_k^Z=0.357\) and \(M_k^X=0.305\), while \method{} improves them to \(M_k^Z=0.372\) and \(M_k^X=0.337\), achieving the best balance between penetration and coverage. The residual spectra in Fig.~\ref{fig6}(c) also show that Text and C3 retain clear anisotropic residual structures, whereas ReAlign and \method{} reduce dominant residual directions. \method{} achieves near-zero centroid discrepancy, the most balanced local support matching, and a much weaker structured anisotropic residual.

\paragraph{\ding{183} \textbf{Does \method{} Preserve Source-Modality Semantics During Modality Alignment?}} 

\begin{wrapfigure}{r}{0.5\textwidth}
    \centering
    \vspace{-10pt}
    \includegraphics[width=\linewidth]{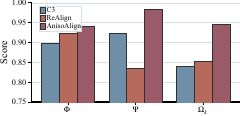}
    \caption{
    Source-modality semantic preservation of different alignment methods. AnisoAlign achieves the best performance across instance consistency, relative geometry consistency, and neighborhood consistency.
    }
    \label{fig7}
\end{wrapfigure}


This experiment evaluates whether modality alignment can preserve the semantic organization of the source modality while performing geometric correction. As shown in Fig.~\ref{fig7}, C3 achieves approximately \(0.899\), \(0.925\), and \(0.840\) on \(\Phi\), \(\Psi\), and \(\Omega_k\), respectively, indicating that Gaussian perturbation can preserve certain global pairwise similarity relations, but introduces noticeable disruption to local neighborhood structures. ReAlign performs well in instance-level consistency, with \(\Phi\approx0.923\), but its relative geometry consistency drops to \(\Psi\approx0.836\), suggesting that pointwise closeness alone does not guarantee the stability of semantic relations within the source modality. In contrast, AnisoAlign achieves the best performance on all three metrics, with \(\Phi\approx0.941\), \(\Psi\approx0.983\), and \(\Omega_k\approx0.945\). This shows that AnisoAlign not only preserves instance-level consistency between transformed representations and original text representations, but also more stably maintains the global semantic relations and local neighborhood structure of the source modality.

\paragraph{\ding{184} \textbf{Can \method{} Improve Fully Text-Only MLLM Training?}}

We next ask whether \method{} can provide an effective visual representation interface for MLLMs without using any image-text pairs throughout training. In this setting, the model cannot learn from real image features and must rely only on substitute visual representations obtained from aligned text representations. All methods use the same protocol: pretraining on Unicorn-1.2M \cite{yu2025unicorntextonlydatasynthesis} followed by instruction-tuning on Unicorn-Instruction-417K \cite{yu2025unicorntextonlydatasynthesis}, with identical data, architecture, and training procedure. As shown in Table~\ref{table1}, \method{} achieves the highest average score, \(47.49\), outperforming ReAlign (\(45.00\)), C3 Align (\(42.44\)), Unicorn (\(42.57\)), and W/o. Align (\(40.08\)). This shows that fully text-only training depends not only on the amount of text data, but also on whether text representations can enter the visual representation space in the correct geometric form. W/o. Align leaves substitute representations near the text distribution; C3 Align and ReAlign alleviate this issue through statistical correction or global distribution matching. In contrast, \method{} jointly models target-modality distribution constraints and source-modality semantic preservation, producing substitute representations better suited as a visual interface.




\definecolor{TableTop}{HTML}{1A365D}   
\definecolor{TableMid}{HTML}{2B6CB0}   
\definecolor{RowAlt}{HTML}{F7FAFC}     
\definecolor{BestRow}{HTML}{FFFBF0}    
\definecolor{BestText}{HTML}{C05621}   

\begin{table*}[t]
\centering
\caption{Results on fully text-only MLLM training setting.}
\label{tab1}
\renewcommand{\arraystretch}{1.3} 
\arrayrulecolor{TableTop} 

\resizebox{\linewidth}{!}{%
\begin{tabular}{l cccc cccc ccc c}
\toprule[1.5pt] 
\multirow{2}{*}{\textcolor{TableTop}{\textbf{Method}}} 
& \multicolumn{4}{c}{\textcolor{TableTop}{\textbf{General}}} 
& \multicolumn{4}{c}{\textcolor{TableTop}{\textbf{Reasoning}}} 
& \multicolumn{3}{c}{\textcolor{TableTop}{\textbf{Hallucination}}} 
& \multirow{2}{*}{\textcolor{TableTop}{\textbf{Avg. $\uparrow$}}} \\
\cmidrule(lr){2-5} \cmidrule(lr){6-9} \cmidrule(lr){10-12}
& \textcolor{TableMid}{\textbf{MME}} & \textcolor{TableMid}{\textbf{MMStar}} & \textcolor{TableMid}{\textbf{SQA}} & \textcolor{TableMid}{\textbf{RWQA}} 
& \textcolor{TableMid}{\textbf{MMMU}} & \textcolor{TableMid}{\textbf{MMMU-P}} & \textcolor{TableMid}{\textbf{VisuLogic}} & \textcolor{TableMid}{\textbf{LogicVista}} 
& \textcolor{TableMid}{\textbf{CRPE}} & \textcolor{TableMid}{\textbf{POPE}} & \textcolor{TableMid}{\textbf{HallBench}} & \\
\midrule[1pt]

Blind       &  3.37 &  8.80 &  6.17 &  5.36 & 19.60 & 12.44 &  0.30 &  1.56 & 12.90 &  0.60 & 15.25 &  7.85 \\

\hdashline

\rowcolor{RowAlt} 
W/o. Align & 46.17 & 30.67 & 58.51 & 37.78 & 30.69 & 29.59 & 25.60 & 24.38 & 65.23 & 55.28 & 37.01 & 40.08  \\
Unicorn & 60.24 & 29.27 & 66.12 & 37.65 & 30.46 & 30.73 & 25.50 & 24.16 & 65.76 & 55.31 & 43.01 & 42.57 \\
C$^3$ Align & 62.56 & 31.40 & 63.30 & 36.47 & 32.67 & 30.34 & 26.00 & 23.27 & 59.07 & 54.17 & 47.63 & 42.44 \\
ReAlign & 67.48 & 32.80 & 65.68 & 40.78 & 33.61 & 31.85 & 26.20 & 25.95 & \textcolor{BestText}{\textbf{67.66}} & 56.91 & 46.06 & 45.00 \\

\hdashline

\rowcolor{BestRow} 
\textbf{\method{}} & 
\textcolor{BestText}{\textbf{72.96}} & 
\textcolor{BestText}{\textbf{34.47}} & \textcolor{BestText}{\textbf{70.81}} & \textcolor{BestText}{\textbf{42.09}} & \textcolor{BestText}{\textbf{37.34}} & \textcolor{BestText}{\textbf{34.05}} & \textcolor{BestText}{\textbf{27.90}} & \textcolor{BestText}{\textbf{27.29}} & 66.36 & \textcolor{BestText}{\textbf{57.62}} & \textcolor{BestText}{\textbf{51.52}} & \textcolor{BestText}{\textbf{47.49}} \\
\bottomrule[1.5pt] 
\end{tabular}%
}
\label{table1}
\vspace{5pt}
\arrayrulecolor{black}
\end{table*}

\paragraph{\ding{185} \textbf{Can \method{} serve as a stronger text-only pretraining interface before visual instruction tuning?}}

\definecolor{TableTop}{HTML}{1A365D}   
\definecolor{TableMid}{HTML}{2B6CB0}   
\definecolor{RowAlt}{HTML}{F7FAFC}     
\definecolor{BestRow}{HTML}{FFFBF0}    
\definecolor{BestText}{HTML}{C05621}   

\begin{table*}[t]
\centering
\vspace{-8pt}
\caption{Results on text-only pretraining setting.}
\label{tab1}
\renewcommand{\arraystretch}{1.3} 
\arrayrulecolor{TableTop} 

\resizebox{\linewidth}{!}{%
\begin{tabular}{l cccc cccc ccc c}
\toprule[1.5pt] 
\multirow{2}{*}{\textcolor{TableTop}{\textbf{Method}}} 
& \multicolumn{4}{c}{\textcolor{TableTop}{\textbf{General}}} 
& \multicolumn{4}{c}{\textcolor{TableTop}{\textbf{Reasoning}}} 
& \multicolumn{3}{c}{\textcolor{TableTop}{\textbf{Hallucination}}} 
& \multirow{2}{*}{\textcolor{TableTop}{\textbf{Avg. $\uparrow$}}} \\
\cmidrule(lr){2-5} \cmidrule(lr){6-9} \cmidrule(lr){10-12}
& \textcolor{TableMid}{\textbf{MME}} & \textcolor{TableMid}{\textbf{MMStar}} & \textcolor{TableMid}{\textbf{SQA}} & \textcolor{TableMid}{\textbf{RWQA}} 
& \textcolor{TableMid}{\textbf{MMMU}} & \textcolor{TableMid}{\textbf{MMMU-P}} & \textcolor{TableMid}{\textbf{VisuLogic}} & \textcolor{TableMid}{\textbf{LogicVista}} 
& \textcolor{TableMid}{\textbf{CRPE}} & \textcolor{TableMid}{\textbf{POPE}} & \textcolor{TableMid}{\textbf{HallBench}} & \\
\midrule[1pt]

Blind       &  3.37 &  8.80 &  6.17 &  5.36 & 19.60 & 12.44 &  0.30 &  1.56 & 12.90 &  0.60 & 15.25 &  7.85 \\

\hdashline

\rowcolor{RowAlt} 
W/o. Align  & 73.63 & 35.73 & 75.23 & 43.53 & 28.82 & 25.38 & 24.40 & 21.03 & 80.82 & 71.59 & 42.38 & 47.50 \\
C$^3$ Align & 76.16 & 34.60 & 75.52 & 43.14 & 30.69 & 27.20 & 25.50 & 19.91 & 79.99 & 72.43 & 43.53 & 48.06 \\
ReAlign & 79.65 & 36.13 & \textcolor{BestText}{\textbf{76.71}} & \textcolor{BestText}{\textbf{47.97}} & 31.51 & 28.39 & 27.70 & 22.82 & 81.78 & 72.53 & 46.58 & 50.16 \\

\hdashline

\rowcolor{BestRow} 
\textbf{\method{}} & \textcolor{BestText}{\textbf{81.22}} & \textcolor{BestText}{\textbf{36.73}} & 76.27 & 44.58 & \textcolor{BestText}{\textbf{37.34}} & \textcolor{BestText}{\textbf{32.85}} & \textcolor{BestText}{\textbf{28.10}} & \textcolor{BestText}{\textbf{25.95}} & \textcolor{BestText}{\textbf{82.93}} & \textcolor{BestText}{\textbf{73.65}} & \textcolor{BestText}{\textbf{47.84}} & \textcolor{BestText}{\textbf{51.59}} \\
\bottomrule[1.5pt] 
\end{tabular}%
}
\vspace{-5pt}
\label{table2}
\arrayrulecolor{black}
\end{table*}

We examine whether \method{} can serve as a stronger text-only pretraining interface before visual instruction tuning. This setting asks whether large-scale text-only data can first be used to construct substitute visual representations during pretraining, followed by post-training with real vision-language instruction data. We use 1M text samples from Bunny-pretrain \cite{he2024efficientmultimodallearningdatacentric} for text-only pretraining and InternVL-Chat-V1.2-SFT \cite{chen2024internvlscalingvisionfoundation} for visual instruction tuning. As shown in Table~\ref{table2}, \method{} achieves the highest average score of 51.59, outperforming ReAlign (50.16), C3 Align (48.06), and W/o. Align (47.50). These results show that \method{} not only provides substitute visual representations in fully text-only training, but also serves as a better pretraining interface before visual instruction tuning. Compared with ReAlign, \method{} improves by 1.43 points, suggesting that global distribution matching alone is insufficient to fully exploit text-only pretraining signals. Its gains over C3 Align and W/o. Align, 3.53 and 4.09 points respectively, further indicate that coarse perturbation or no explicit alignment cannot construct a stable visual substitute interface.

\definecolor{TableTop}{HTML}{1A365D}   
\definecolor{TableMid}{HTML}{2B6CB0}   
\definecolor{RowAlt}{HTML}{F7FAFC}     
\definecolor{BestRow}{HTML}{FFFBF0}    
\definecolor{BestText}{HTML}{C05621}   

\begin{table*}[t]
\centering
\caption{Scaling text-only data with AnisoAlign enables substitute visual representations to approach and slightly surpass real image-based pretraining.}
\label{tab1}

\renewcommand{\arraystretch}{1.3} 
\arrayrulecolor{TableTop} 

\resizebox{\linewidth}{!}{%
\begin{tabular}{l cccc cccc ccc c}
\toprule[1.5pt] 
\multirow{2}{*}{\textcolor{TableTop}{\textbf{Method}}} 
& \multicolumn{4}{c}{\textcolor{TableTop}{\textbf{General}}} 
& \multicolumn{4}{c}{\textcolor{TableTop}{\textbf{Reasoning}}} 
& \multicolumn{3}{c}{\textcolor{TableTop}{\textbf{Hallucination}}} 
& \multirow{2}{*}{\textcolor{TableTop}{\textbf{Avg. $\uparrow$}}} \\
\cmidrule(lr){2-5} \cmidrule(lr){6-9} \cmidrule(lr){10-12}
& \textcolor{TableMid}{\textbf{MME}} & \textcolor{TableMid}{\textbf{MMStar}} & \textcolor{TableMid}{\textbf{SQA}} & \textcolor{TableMid}{\textbf{RWQA}} 
& \textcolor{TableMid}{\textbf{MMMU}} & \textcolor{TableMid}{\textbf{MMMU-P}} & \textcolor{TableMid}{\textbf{VisuLogic}} & \textcolor{TableMid}{\textbf{LogicVista}} 
& \textcolor{TableMid}{\textbf{CRPE}} & \textcolor{TableMid}{\textbf{POPE}} & \textcolor{TableMid}{\textbf{HallBench}} & \\
\midrule[1pt]

W/. Image   & 82.86 & 37.07 & 77.67 & 45.62 & 38.27 & 33.04 & \textcolor{BestText}{\textbf{29.40}} & 27.08 & 82.73 & 74.16 & 48.06 & 52.72 \\

\rowcolor{RowAlt}
\textbf{\method{}-1M} & 81.22 & 36.73 & 76.27 & 44.58 & 37.34 & 32.85 & 28.10 & 25.95 & \textcolor{BestText}{\textbf{82.93}} & 73.65 & 47.84 & 51.60 \\

\textbf{\method{}-2M} 

& \textcolor{BestText}{\textbf{83.15}} 
& \textcolor{BestText}{\textbf{37.47}} 
& \textcolor{BestText}{\textbf{78.60}} 
& \textcolor{BestText}{\textbf{45.79}} 
& \textcolor{BestText}{\textbf{38.92}} 
& \textcolor{BestText}{\textbf{33.86}} 
& 29.20 
& \textcolor{BestText}{\textbf{27.64}} 
& 82.17 
& \textcolor{BestText}{\textbf{75.39}} 
& \textcolor{BestText}{\textbf{49.63}} 
& \textcolor{BestText}{\textbf{52.75}} \\

\bottomrule[1.5pt] 
\end{tabular}%
}
\label{tab:text_only_scaling}
\arrayrulecolor{black}
\end{table*}

\paragraph{\ding{186} \textbf{Can \method{} surpass paired image-text pretraining by scaling up text-only data?}} We examine a further question: if the scale of text-only data continues to increase, can \method{} approach or even surpass pretraining with real image-text pairs? This experiment is designed to verify the scalability of \method{}. If the quality of substitute visual representations is sufficiently high, then text-only data can not only serve as a supplement to real image data, but may also become a more economical and more scalable pretraining resource in large-scale scenarios. We compare three settings: (1) W/. Image, which uses real images; (2) \method{}-1M, which uses 1M text-only samples; and (3) \method{}-2M, which uses 2M text-only samples.
All methods then follow the same downstream training and evaluation pipeline. Table~\ref{tab:text_only_scaling} shows that \method{}-1M already reaches an average score of 51.60, close to W/. Image at 52.72. When the text-only data scale is increased to 2M, \method{}-2M further improves to 52.75, slightly surpassing W/. Image at 52.72 and improving over \method{}-1M by 1.15 average points.
This indicates that real images are not the only scalable source of visual supervision for MLLM pretraining. \method{} provides a scalable training paradigm: through high-quality anisotropic modality alignment, large-scale text data can be transformed into effective visual-style training signals, and can partially reach or even surpass the performance of real image-text pretraining.

\paragraph{\ding{187} \textbf{Are All Components Necessary for \method{}?}} We conduct ablation studies to analyze the contribution of each component in AnisoAlign. As shown in Table.~\ref{tab:ablation}. Using only global initialization achieves an average score of 43.59, showing that coarse centroid and distribution calibration already provide a reasonable substitute representation, but remain insufficient for high-quality modality alignment. Adding instance-conditioned refinement improves the average score to 44.93, indicating that bounded sample-specific correction is necessary beyond global statistics. Introducing the target-modality prior \(L_G\) further raises the score to 46.56, demonstrating that target-side geometric guidance helps the substitute representations better match the visual distribution. Similarly, adding the relative phase constraint \(L_\Phi\) achieves 46.45, confirming the importance of preserving structured phase relations during refinement. The full AnisoAlign model obtains the best average score of 47.49, outperforming all ablated variants and achieving consistent gains across general perception, reasoning, and hallucination-related benchmarks. These results show that global initialization, bounded refinement, target-prior guidance, and phase-structure preservation are complementary and jointly contribute to more effective anisotropic modality alignment.

\begin{table*}[t]
\centering
\caption{Ablation results show that global initialization, bounded refinement, target-prior guidance, and phase-structure preservation each contribute to the final performance of AnisoAlign.}
\label{tab:ablation}
\renewcommand{\arraystretch}{1.3}
\arrayrulecolor{TableTop}

\resizebox{\linewidth}{!}{%
\begin{tabular}{l cccc cccc ccc c}
\toprule[1.5pt]
\multirow{2}{*}{\textcolor{TableTop}{\textbf{Method}}}
& \multicolumn{4}{c}{\textcolor{TableTop}{\textbf{General}}}
& \multicolumn{4}{c}{\textcolor{TableTop}{\textbf{Reasoning}}}
& \multicolumn{3}{c}{\textcolor{TableTop}{\textbf{Hallucination}}}
& \multirow{2}{*}{\textcolor{TableTop}{\textbf{Avg. $\uparrow$}}} \\
\cmidrule(lr){2-5} \cmidrule(lr){6-9} \cmidrule(lr){10-12}
& \textcolor{TableMid}{\textbf{MME}} 
& \textcolor{TableMid}{\textbf{MMStar}} 
& \textcolor{TableMid}{\textbf{SQA}} 
& \textcolor{TableMid}{\textbf{RWQA}}
& \textcolor{TableMid}{\textbf{MMMU}} 
& \textcolor{TableMid}{\textbf{MMMU-P}} 
& \textcolor{TableMid}{\textbf{VisuLogic}} 
& \textcolor{TableMid}{\textbf{LogicVista}}
& \textcolor{TableMid}{\textbf{CRPE}} 
& \textcolor{TableMid}{\textbf{POPE}} 
& \textcolor{TableMid}{\textbf{HallBench}} 
& \\
\midrule[1pt]

Global Initialization Only 
& 64.32 & 32.10 & 64.85 & 39.66 
& 33.02 & 31.41 & 25.90 & 25.48 
& 62.19 & 53.36 & 47.21 
& 43.59 \\

\rowcolor{RowAlt}
Global Initialization + Refinement 
& 67.85 & 32.63 & 66.94 & 40.57 
& 34.18 & 32.16 & 26.41 & 25.97 
& 62.65 & 55.71 & 49.18 
& 44.93 \\

Global Initialization + Refinement + $\mathcal{L}^{\mathrm{II}}$ 
& 71.24 & 33.86 & 68.47 & 41.91 
& 35.63 & 33.41 & 27.03 & 26.52 
& 64.94 & 56.28 & 52.83 
& 46.56 \\

\rowcolor{RowAlt}
Global Initialization + Refinement + $\mathcal{L}^{\Phi}$ 
& 70.38 & 33.41 & 69.22 & 41.67 
& 35.97 & 32.89 & 27.46 & 26.91 
& 65.32 & 55.97 & 51.74 
& 46.45 \\

\rowcolor{BestRow}
\textbf{\method{}}
& \textcolor{BestText}{\textbf{72.96}} 
& \textcolor{BestText}{\textbf{34.47}} 
& \textcolor{BestText}{\textbf{70.81}} 
& \textcolor{BestText}{\textbf{42.09}}
& \textcolor{BestText}{\textbf{37.34}} 
& \textcolor{BestText}{\textbf{34.05}} 
& \textcolor{BestText}{\textbf{27.90}} 
& \textcolor{BestText}{\textbf{27.29}}
& \textcolor{BestText}{\textbf{66.36}}
& \textcolor{BestText}{\textbf{57.62}} 
& \textcolor{BestText}{\textbf{51.52}} 
& \textcolor{BestText}{\textbf{47.49}} \\

\bottomrule[1.5pt]
\end{tabular}%
}
\label{tab:ablation}
\arrayrulecolor{black}
\end{table*}

\section{Conclusion}

This paper revisits the modality gap from a geometric perspective and shows that it is a structured anisotropic residual built upon compatible semantic geometry. Based on this observation, we propose the principle of anisotropic modality alignment. We further propose an unpaired modality alignment method for generating target-modality substitute representations. Experiments show that our method can help eliminate reliance on paired image-text data. Overall, modality alignment should be better understood as a structured geometric correction.

\newpage

\bibliography{ref}
\bibliographystyle{plainnat}


\newpage
\appendix
\onecolumn

\appendix



\section{Theoretical Derivation of the Anisotropic Modality Gap}
\label{app:1}

\subsection{Overview and Notation}
\label{app:overview-notation}

This appendix provides theoretical support for the geometric diagnostics in Sec.~3 and the methodological design in Sec.~4. The objective is not to prove the global optimality of the proposed alignment algorithm, but to formalize the following four points: first, why the modality gap should be decomposed into a global centroid displacement and a centered residual component; second, why the centered residual should be compared against an isotropic null hypothesis; third, why the dominant residual directions constitute efficient correction targets for reducing the residual gap; and fourth, why effective correction must be constrained in order to simultaneously preserve the semantic geometry of the source modality and improve compatibility with the target-modality distribution.

Let \(X\) denote the target image modality and \(Y\) denote the source text modality. Let \(x\in X\) and \(y\in Y\) be a paired image-text representation in the shared embedding space \(\mathbb{R}^{d}\). Paired samples are used only for the geometric diagnostics in Sec.~3; the alignment method proposed in Sec.~4 does not rely on paired supervision. We denote the modality means by
\[
\mu_x \coloneqq \mathbb{E}[x],
\qquad
\mu_y \coloneqq \mathbb{E}[y],
\]
and define the centered variables as
\[
\bar{x} \coloneqq x-\mu_x,
\qquad
\bar{y} \coloneqq y-\mu_y.
\]
The centered covariance matrices of the two modalities are
\[
\Sigma_x \coloneqq \mathbb{E}[\bar{x}\bar{x}^{\top}],
\qquad
\Sigma_y \coloneqq \mathbb{E}[\bar{y}\bar{y}^{\top}].
\]
The centered cross-modal second-order moments are denoted by
\[
\Sigma_{xy} \coloneqq \mathbb{E}[\bar{x}\bar{y}^{\top}],
\qquad
\Sigma_{yx} \coloneqq \Sigma_{xy}^{\top}.
\]
For any symmetric matrix \(M\), let \(\lambda_j(M)\) denote its \(j\)-th largest eigenvalue in descending order, let \(U_M^q\) denote the matrix formed by its top \(q\) eigenvectors, and let \(P_M^q \coloneqq U_M^q(U_M^q)^{\top}\) denote the corresponding orthogonal projector.

\subsection{Formalizing Dominant Geometric Compatibility}
\label{app:dominant-geometric-compatibility}

Before studying how to correct the modality gap, we first ask whether the two modalities possess compatible global geometry in the shared representation space. If image and text representations were two arbitrary and unrelated distributions, then any geometric correction would be unlikely to simultaneously achieve distributional alignment and semantic preservation.

\subsubsection{Spectral Compatibility}
\label{app:spectral-compatibility}

Let the eigenvalues of \(\Sigma_x\) and \(\Sigma_y\) be sorted in descending order. We measure whether the two modalities allocate variance energy similarly across dominant and tail directions using the log-spectral correlation
\[
C_{\lambda}
\coloneqq
\mathrm{corr}\!\left(
\log \lambda(\Sigma_x),
\log \lambda(\Sigma_y)
\right).
\]
A high value of \(C_{\lambda}\) indicates that the two modalities exhibit similar hierarchical variance-energy profiles. In other words, if one modality allocates substantial variance to certain dominant directions, the other modality tends to allocate substantial variance to directions at similar spectral ranks.

However, spectral similarity alone does not imply that the two modalities use the same directions. Two covariance matrices may have similar eigenvalue spectra while having nearly orthogonal eigenspaces. Therefore, we further compare their principal subspaces.

\subsubsection{Principal Subspace Overlap and the Random Baseline}
\label{app:principal-subspace-overlap}

Let \(U_x^q\) and \(U_y^q\) denote the top-\(q\) eigenvector matrices of \(\Sigma_x\) and \(\Sigma_y\), respectively. We define the principal subspace overlap as
\[
O_q
\coloneqq
\frac{1}{q}
\left\|
(U_x^q)^{\top}U_y^q
\right\|_F^2
=
\frac{1}{q}\mathrm{tr}(P_x^qP_y^q),
\]
where \(P_x^q=U_x^q(U_x^q)^{\top}\) and \(P_y^q=U_y^q(U_y^q)^{\top}\). This quantity lies in \([0,1]\) and measures the degree of overlap between the two \(q\)-dimensional principal subspaces.

\begin{lemma}[Random-subspace baseline]
\label{lem:random-subspace-baseline}
Suppose \(U_y^q\) is sampled uniformly from the Grassmann manifold \(\mathrm{Gr}(q,d)\) with respect to the Haar measure and is independent of \(U_x^q\). Then
\[
\mathbb{E}[P_y^q]=\frac{q}{d}I_d,
\qquad
\mathbb{E}[O_q]=\frac{q}{d}.
\]
\end{lemma}

\begin{proof}
By the invariance of the Haar measure under the action of the orthogonal group, \(\mathbb{E}[P_y^q]\) must commute with every orthogonal matrix. Hence it must be a scalar multiple of the identity matrix. Since \(\mathrm{tr}(P_y^q)=q\), we obtain \(\mathbb{E}[P_y^q]=(q/d)I_d\). Therefore,
\[
\mathbb{E}[O_q]
=
\frac{1}{q}
\mathrm{tr}\!\left(P_x^q\mathbb{E}[P_y^q]\right)
=
\frac{1}{q}
\mathrm{tr}\!\left(P_x^q\frac{q}{d}I_d\right)
=
\frac{q}{d}.
\]
\end{proof}

Thus, when the empirical overlap satisfies \(O_q \gg q/d\), the two modalities do not use randomly unrelated dominant directions. Instead, they share a non-random set of dominant geometric directions.

\subsubsection{Implication}
\label{app:dominant-geometry-implication}

Spectral compatibility and principal subspace overlap together indicate that the modality gap is not an arbitrary discrepancy between two unrelated distributions. Rather, it is a structured discrepancy built upon a partially shared dominant semantic-geometric backbone. This observation provides a necessary premise for alignment: the transformation should not freely distort the source-modality representation, but should preserve its existing semantic geometry while correcting the residual structure that prevents compatibility with the target modality.

\subsection{Mean--Residual Decomposition}
\label{app:mean-residual-decomposition}

\subsubsection{Decomposition Identity}
\label{app:decomposition-identity}

For a paired representation \((x,y)\), we have
\[
x-y=(\mu_x-\mu_y)+(\bar{x}-\bar{y}).
\]
Since \(\mathbb{E}[\bar{x}-\bar{y}]=0\), it follows that
\[
\mathbb{E}\!\left[
\left\langle
\mu_x-\mu_y,\,
\bar{x}-\bar{y}
\right\rangle
\right]=0.
\]
Therefore, the expected squared cross-modal discrepancy admits the orthogonal decomposition
\[
\mathbb{E}\|x-y\|_2^2
=
\|\mu_x-\mu_y\|_2^2
+
\mathbb{E}\|\bar{x}-\bar{y}\|_2^2.
\]
This decomposition shows that the modality gap contains at least two components: the first-order centroid displacement \(\|\mu_x-\mu_y\|_2^2\) and the centered residual discrepancy \(\mathbb{E}\|\bar{x}-\bar{y}\|_2^2\). Therefore, global mean displacement can only explain first-order centroid mismatch, but not the structured discrepancy that remains after centering.

\subsubsection{Residual after Centroid Correction}
\label{app:residual-after-centroid-correction}

Consider the global centroid correction applied to the text representation,
\[
y^x \coloneqq y-\mu_y+\mu_x.
\]
The paired residual after this correction is
\[
r
\coloneqq
x-y^x
=
(x-\mu_x)-(y-\mu_y)
=
\bar{x}-\bar{y}.
\]
Thus \(\mathbb{E}[r]=0\), and its covariance matrix is
\[
\Sigma_r
\coloneqq
\mathbb{E}[rr^{\top}].
\]
Expanding this expression gives
\[
\Sigma_r
=
\Sigma_x+\Sigma_y-\Sigma_{xy}-\Sigma_{yx}.
\]
This identity shows that the centered residual is determined not only by the marginal covariance structures of the two modalities, but also by their cross-modal correspondence structure \(\Sigma_{xy}\).

The squared residual energy remaining after centroid correction is
\[
\mathbb{E}\|r\|_2^2=\mathrm{tr}(\Sigma_r).
\]
Therefore, if the modality gap were mainly dominated by a global centroid displacement, then \(\mathrm{tr}(\Sigma_r)\) should become small after centroid correction. Conversely, if the residual remains large, then the modality gap cannot be explained as a simple global mean shift.

\subsubsection{Residual Ratio}
\label{app:residual-ratio}

To compare across datasets or embedding scales, we may define the energy-based residual ratio as
\[
R_{\mathrm{energy}}
\coloneqq
\frac{\mathbb{E}\|r\|_2^2}{\mathbb{E}\|x-y\|_2^2}
=
\frac{\mathrm{tr}(\Sigma_r)}
{\|\mu_x-\mu_y\|_2^2+\mathrm{tr}(\Sigma_r)}.
\]
If the main text reports average distances rather than squared energies, the corresponding distance-based ratio can be defined as
\[
R_{\mathrm{dist}}
\coloneqq
\frac{\mathbb{E}\|r\|_2}{\mathbb{E}\|x-y\|_2}.
\]
Under either definition, a ratio close to \(1\) indicates that most of the cross-modal discrepancy remains in the centered residual after removing the global centroid displacement. The large residual ratio observed in the main text therefore rejects the simple explanation that the modality gap is primarily a centroid bias.

\subsection{Isotropic Residual Null Hypothesis}
\label{app:isotropic-residual-null}

After removing the centroid displacement, a natural null hypothesis is that the remaining residual is merely isotropic noise. Under this hypothesis, the residual has equal variance in all directions and does not contain any dominant geometric structure.

\subsubsection{Null Hypothesis}
\label{app:null-hypothesis}

We formalize the isotropic residual null hypothesis as
\[
\mathcal{H}_0:\qquad \Sigma_r=\sigma^2 I_d
\]
for some \(\sigma>0\). Under \(\mathcal{H}_0\), all eigenvalues of \(\Sigma_r\) are equal:
\[
\lambda_1(\Sigma_r)=\lambda_2(\Sigma_r)=\cdots=\lambda_d(\Sigma_r)=\sigma^2.
\]
This hypothesis implies three direct spectral properties.

\subsubsection{Residual Anisotropy Ratio}
\label{app:anisotropy-ratio}

We define the residual anisotropy ratio as
\[
A_r
\coloneqq
\frac{\lambda_{\max}(\Sigma_r)}
{\mathrm{tr}(\Sigma_r)/d},
\]
where \(\lambda_{\max}(\Sigma_r)\) is the largest eigenvalue of the residual covariance and \(\mathrm{tr}(\Sigma_r)/d\) is the average eigenvalue. Since the largest eigenvalue is no smaller than the average eigenvalue, \(A_r\geq 1\). Under the isotropic null hypothesis, all eigenvalues are equal, hence \(A_r=1\). Therefore, an empirical observation of \(A_r\gg 1\) indicates that some directions carry residual energy far above the average level, contradicting the isotropic-noise hypothesis.

\subsubsection{Cumulative Spectral Energy}
\label{app:cumulative-spectral-energy}

We define the cumulative energy explained by the top \(K\) residual eigen-directions as
\[
E(K)
\coloneqq
\frac{\sum_{j=1}^{K}\lambda_j(\Sigma_r)}
{\sum_{j=1}^{d}\lambda_j(\Sigma_r)}.
\]
Under the isotropic null hypothesis,
\[
E(K)=\frac{K}{d}.
\]
Thus, if the empirical curve satisfies \(E(K)\gg K/d\) for small \(K\), the residual energy is concentrated in a small number of dominant directions.

\subsubsection{Effective Dimension}
\label{app:effective-dimension}

The effective dimension of the residual covariance is defined as
\[
d_{\mathrm{eff}}(\Sigma_r)
\coloneqq
\frac{\mathrm{tr}(\Sigma_r)^2}{\mathrm{tr}(\Sigma_r^2)}
=
\frac{\left(\sum_j \lambda_j(\Sigma_r)\right)^2}
{\sum_j \lambda_j(\Sigma_r)^2}.
\]
By the Cauchy--Schwarz inequality, \(1\leq d_{\mathrm{eff}}(\Sigma_r)\leq d\). Under the isotropic null hypothesis, \(d_{\mathrm{eff}}(\Sigma_r)=d\). Hence, an empirical observation of \(d_{\mathrm{eff}}(\Sigma_r)/d\ll 1\) indicates that the residual distribution has an effective dimension much smaller than the ambient dimension.

\subsubsection{Implication}
\label{app:isotropic-null-implication}

The isotropic residual null hypothesis is rejected by any of the following empirical patterns:
\[
A_r\gg 1,
\qquad
E(K)\gg \frac{K}{d},
\qquad
\frac{d_{\mathrm{eff}}(\Sigma_r)}{d}\ll 1.
\]
The residual spectrum reported in the main text satisfies these conditions simultaneously. This indicates that the centered residual is not unstructured isotropic noise, but a structured anisotropic residual concentrated along a small number of dominant directions.

\subsection{Efficiency of Dominant Residual-Direction Correction}
\label{app:dominant-residual-correction}

The previous subsection shows that the centered residual energy is concentrated along a small number of dominant directions. We now show that if a correction is restricted to act within a \(K\)-dimensional subspace, then choosing the top \(K\) eigen-directions of the residual covariance is optimal for minimizing the remaining squared residual energy.

\subsubsection{Optimal Projection Result}
\label{app:optimal-projection}

Let the eigendecomposition of the residual covariance be
\[
\Sigma_r=U\Lambda U^{\top},
\qquad
\Lambda=\mathrm{diag}(\lambda_1,\ldots,\lambda_d),
\]
where \(\lambda_1\geq\lambda_2\geq\cdots\geq\lambda_d\geq 0\). Consider an oracle correction that removes the residual component in some \(K\)-dimensional subspace. Let \(P\) be the orthogonal projector onto that subspace. The corrected residual is \((I-P)r\), and the expected remaining residual energy is
\[
J(P)
\coloneqq
\mathbb{E}\|(I-P)r\|_2^2.
\]
Since
\[
J(P)
=
\mathrm{tr}\!\left((I-P)\Sigma_r\right)
=
\mathrm{tr}(\Sigma_r)-\mathrm{tr}(P\Sigma_r),
\]
minimizing \(J(P)\) is equivalent to maximizing \(\mathrm{tr}(P\Sigma_r)\).

\begin{proposition}[Optimal rank-constrained residual correction]
\label{prop:optimal-residual-subspace}
Among all rank-\(K\) orthogonal projectors, the projector onto the subspace spanned by the top \(K\) eigenvectors of \(\Sigma_r\) minimizes the expected remaining residual energy:
\[
P_K^{\star}
=
\arg\min_{\mathrm{rank}(P)=K}
\mathbb{E}\|(I-P)r\|_2^2.
\]
The minimum value is
\[
\min_{\mathrm{rank}(P)=K}
\mathbb{E}\|(I-P)r\|_2^2
=
\sum_{j>K}\lambda_j(\Sigma_r).
\]
\end{proposition}

\begin{proof}
By the Ky Fan maximum principle,
\[
\max_{\mathrm{rank}(P)=K}\mathrm{tr}(P\Sigma_r)
=
\sum_{j=1}^{K}\lambda_j(\Sigma_r),
\]
and the maximum is attained by the projector onto the top-\(K\) eigenspace of \(\Sigma_r\). Substituting this into \(J(P)=\mathrm{tr}(\Sigma_r)-\mathrm{tr}(P\Sigma_r)\) yields
\[
\min_{\mathrm{rank}(P)=K}J(P)
=
\mathrm{tr}(\Sigma_r)-\sum_{j=1}^{K}\lambda_j(\Sigma_r)
=
\sum_{j>K}\lambda_j(\Sigma_r).
\]
\end{proof}

\subsubsection{Comparison with Random Correction}
\label{app:random-correction-comparison}

If the rank-\(K\) correction subspace is chosen randomly, with projector \(P_{\mathrm{rand}}\), then
\[
\mathbb{E}[P_{\mathrm{rand}}]=\frac{K}{d}I_d.
\]
Therefore, the expected residual energy removed by a random subspace is
\[
\mathbb{E}\!\left[\mathrm{tr}(P_{\mathrm{rand}}\Sigma_r)\right]
=
\frac{K}{d}\mathrm{tr}(\Sigma_r).
\]
By contrast, the top \(K\) residual eigen-directions remove
\[
\sum_{j=1}^{K}\lambda_j(\Sigma_r)
=
E(K)\mathrm{tr}(\Sigma_r).
\]
The ratio between dominant-direction correction and random correction is
\[
\frac{\sum_{j=1}^{K}\lambda_j(\Sigma_r)}
{(K/d)\mathrm{tr}(\Sigma_r)}
=
\frac{E(K)}{K/d}.
\]
When \(E(K)\gg K/d\), correcting the dominant residual directions is substantially more efficient than correcting random directions or applying isotropic perturbations. In particular, when \(K=1\), this gain is exactly
\[
A_r
=
\frac{\lambda_{\max}(\Sigma_r)}
{\mathrm{tr}(\Sigma_r)/d}.
\]
Thus, stronger residual anisotropy implies a larger advantage for dominant-direction correction over random correction.

\subsubsection{From the Residual Principal Subspace to the Joint Covariance Subspace}
\label{app:joint-covariance-surrogate}

The oracle analysis above indicates that, if paired residuals were available, the most direct correction target would be the principal subspace of \(\Sigma_r\). However, in the unpaired alignment setting, estimating \(\Sigma_r\) requires cross-modal correspondence through
\[
\Sigma_r=\Sigma_x+\Sigma_y-\Sigma_{xy}-\Sigma_{yx}.
\]
The cross-modal term \(\Sigma_{xy}\) cannot be reliably estimated from unpaired samples alone. Therefore, directly using the residual covariance to define the correction subspace is not suitable for unpaired alignment.

The proposed method instead uses the joint marginal covariance
\[
\Sigma_J
\coloneqq
\Sigma_x+\Sigma_y+\lambda I
\]
and takes its top \(r\) eigenvectors to define the dominant subspace \(U\). This choice should not be interpreted as a claim that the principal subspace of \(\Sigma_J\) is strictly identical to that of \(\Sigma_r\). Rather, \(\Sigma_J\) provides a computable unpaired surrogate based only on marginal statistics.

This surrogate is motivated by two observations. First, the spectral compatibility and principal subspace overlap in Sec.~\ref{app:dominant-geometric-compatibility} indicate that image and text representations already share a non-random dominant geometric backbone. Hence, the principal subspace of \(\Sigma_x+\Sigma_y\) captures high-variance geometric directions jointly occupied by both modalities. Second, the residual spectral diagnostics in Sec.~\ref{app:isotropic-residual-null} show that the remaining modality gap is not uniformly distributed over all directions, but concentrated in a low-effective-dimensional structure. Therefore, applying structured correction within the shared dominant geometric backbone is more consistent with the residual geometry than either unconstrained full-space mappings or isotropic perturbations.

To empirically verify whether the joint dominant subspace captures residual energy, one may report the residual-energy coverage ratio
\[
\eta_U
\coloneqq
\frac{\mathrm{tr}(P_U\Sigma_r)}
{\mathrm{tr}(\Sigma_r)},
\]
where \(P_U\) is the orthogonal projector onto the top-\(r\) eigenspace of \(\Sigma_J\). A high value of \(\eta_U\) indicates that the joint dominant subspace captures a large fraction of the residual energy, further supporting its use as an unpaired surrogate correction subspace.

\subsubsection{Implication}
\label{app:dominant-correction-implication}

This subsection yields two conclusions. First, in the oracle setting where the residual covariance is available, correcting the dominant residual directions is optimal for reducing squared residual energy under a rank constraint. Second, in the practical unpaired setting, where the residual covariance cannot be directly constructed, the dominant eigenspace of the joint marginal covariance provides a computable surrogate. Its use is motivated by the observed dominant geometric compatibility and can be further validated by the residual-energy coverage ratio \(\eta_U\).

\subsection{Non-identifiability of Distribution Matching Alone}
\label{app:nonidentifiability}

The previous subsection clarifies which directions should be corrected. We now explain why matching the target distribution alone is insufficient.

Let \(P_Y\) and \(P_X\) denote the source- and target-modality representation distributions, respectively. A distribution-matching alignment seeks a map \(T\) such that
\[
T_{\#}P_Y=P_X,
\]
where \(T_{\#}P_Y\) denotes the pushforward distribution of \(P_Y\) under \(T\). However, this condition alone does not identify a semantics-preserving alignment map.

\begin{proposition}[Non-identifiability of marginal distribution matching]
\label{prop:nonidentifiability}
Suppose \(T_0\) satisfies \((T_0)_{\#}P_Y=P_X\). For any measurable transformation \(S\) that preserves the target distribution, i.e., \(S_{\#}P_X=P_X\), the composite map \(S\circ T_0\) also satisfies \((S\circ T_0)_{\#}P_Y=P_X\).
\end{proposition}

\begin{proof}
By the composition property of pushforward measures,
\[
(S\circ T_0)_{\#}P_Y
=
S_{\#}\!\left((T_0)_{\#}P_Y\right)
=
S_{\#}P_X
=
P_X.
\]
\end{proof}

Therefore, \(S\circ T_0\) is equally valid as \(T_0\) under marginal distribution matching. However, different choices of \(S\) may arbitrarily permute or distort instance-level semantic correspondence while preserving the same target marginal distribution. Hence, target distribution matching alone cannot distinguish a semantics-preserving alignment from a semantically destructive transformation with the correct marginal distribution.

This explains the role of the random target replacement \(T_{\mathrm{perm}}\) in the main text. It can match the target-modality distribution, but destroys the semantic correspondence between the original source sample and its transformed representation. Effective modality alignment must therefore impose, either explicitly or implicitly, additional constraints that preserve the semantic geometry of the source modality.

\subsection{Semantic Preservation under Bounded Correction}
\label{app:bounded-correction}

The previous subsection shows that target distribution matching alone does not guarantee semantic preservation. We now show that bounded correction controls the distortion of source-modality semantic structure.

\subsubsection{Similarity Preservation under Additive Perturbation}
\label{app:additive-perturbation}

Let
\[
T(y)=y+\delta(y),
\]
where \(\delta(y)\) is the correction applied to the source representation. Assume the source representation is normalized, i.e., \(\|y\|_2=1\), and the correction satisfies \(\|\delta(y)\|_2\leq\varepsilon\). For two source samples \(y_i\) and \(y_j\), let
\[
z_i=y_i+\delta_i,
\qquad
z_j=y_j+\delta_j,
\]
where \(\delta_i=\delta(y_i)\) and \(\delta_j=\delta(y_j)\).

\begin{lemma}[Similarity stability under bounded correction]
\label{lem:bounded-correction}
If \(\|y_i\|_2=\|y_j\|_2=1\) and \(\|\delta_i\|_2,\|\delta_j\|_2\leq\varepsilon\), then
\[
\left|
\langle z_i,z_j\rangle
-
\langle y_i,y_j\rangle
\right|
\leq
2\varepsilon+\varepsilon^2.
\]
If, in addition, \(\varepsilon<1\) and \(\hat{z}_i=z_i/\|z_i\|_2\), then
\[
\|\hat{z}_i-y_i\|_2
\leq
\frac{2\varepsilon}{1-\varepsilon},
\qquad
\left|
\langle \hat{z}_i,\hat{z}_j\rangle
-
\langle y_i,y_j\rangle
\right|
\leq
\frac{4\varepsilon}{1-\varepsilon}.
\]
\end{lemma}

\begin{proof}
For the unnormalized representations,
\[
\begin{aligned}
\left|
\langle z_i,z_j\rangle
-
\langle y_i,y_j\rangle
\right|
&=
\left|
\langle y_i+\delta_i,y_j+\delta_j\rangle
-
\langle y_i,y_j\rangle
\right| \\
&=
\left|
\langle \delta_i,y_j\rangle
+
\langle y_i,\delta_j\rangle
+
\langle \delta_i,\delta_j\rangle
\right| \\
&\leq
|\langle \delta_i,y_j\rangle|
+
|\langle y_i,\delta_j\rangle|
+
|\langle \delta_i,\delta_j\rangle| \\
&\leq
\varepsilon+\varepsilon+\varepsilon^2
=
2\varepsilon+\varepsilon^2.
\end{aligned}
\]
For the normalized representation, \(\|z_i\|_2=\|y_i+\delta_i\|_2\geq 1-\varepsilon\). Therefore,
\[
\begin{aligned}
\|\hat{z}_i-y_i\|_2
&=
\left\|
\frac{y_i+\delta_i}{\|z_i\|_2}
-
y_i
\right\|_2 \\
&\leq
\frac{\|\delta_i\|_2}{\|z_i\|_2}
+
\left|
\frac{1}{\|z_i\|_2}-1
\right|\|y_i\|_2 \\
&\leq
\frac{\varepsilon}{1-\varepsilon}
+
\frac{\varepsilon}{1-\varepsilon}
=
\frac{2\varepsilon}{1-\varepsilon}.
\end{aligned}
\]
Finally,
\[
\begin{aligned}
\left|
\langle \hat{z}_i,\hat{z}_j\rangle
-
\langle y_i,y_j\rangle
\right|
&=
\left|
\langle \hat{z}_i-y_i,\hat{z}_j\rangle
+
\langle y_i,\hat{z}_j-y_j\rangle
\right| \\
&\leq
\|\hat{z}_i-y_i\|_2\|\hat{z}_j\|_2
+
\|y_i\|_2\|\hat{z}_j-y_j\|_2 \\
&\leq
\frac{4\varepsilon}{1-\varepsilon}.
\end{aligned}
\]
\end{proof}

Thus, as long as the Euclidean norm of the correction is controlled, the pairwise inner-product structure of the source modality cannot be arbitrarily distorted.

\subsubsection{Connection to Stage-II Bounded Residual Refinement}
\label{app:stage-ii-bound}

Stage II in the proposed method does not predict an unconstrained free mapping. Instead, it performs bounded correction on the phase, radius, and \(V\)-subspace components:
\[
\hat{\theta}
=
\mathrm{wrap}\!\left(
\theta^{(0)}
+
\alpha_{\theta}\tanh(\Delta\theta)
\right),
\]
\[
\hat{\rho}_k
=
\rho_k^{(0)}
\exp\!\left(
\alpha_{\rho}\tanh(\Delta\rho_k)
\right),
\]
\[
\hat{v}
=
v^{(0)}
+
\alpha_v\tanh(\Delta v_V).
\]
Since \(|\tanh(\cdot)|\leq 1\), each type of correction is explicitly bounded:
\[
|\hat{\theta}_k-\theta_k^{(0)}|
\leq
\alpha_{\theta},
\qquad
\frac{\hat{\rho}_k}{\rho_k^{(0)}}
\in
[e^{-\alpha_{\rho}},e^{\alpha_{\rho}}],
\]
and the \(V\)-side correction is controlled by \(\alpha_v\).

For the \(k\)-th two-dimensional polar block, let the initialized Cartesian coordinate be
\[
c_k^{(0)}
=
\rho_k^{(0)}
(\cos\theta_k^{(0)},\sin\theta_k^{(0)}),
\]
and the updated coordinate be
\[
\hat{c}_k
=
\hat{\rho}_k
(\cos\hat{\theta}_k,\sin\hat{\theta}_k).
\]
Let \(s_k\coloneqq \hat{\rho}_k/\rho_k^{(0)}\) and let \(\Delta_k\coloneqq \hat{\theta}_k-\theta_k^{(0)}\) denote the wrapped angular difference. Then \(s_k\in[e^{-\alpha_{\rho}},e^{\alpha_{\rho}}]\) and \(|\Delta_k|\leq \alpha_{\theta}\). In complex notation,
\[
\frac{\hat{c}_k-c_k^{(0)}}{\rho_k^{(0)}}
=
e^{i\theta_k^{(0)}}(s_ke^{i\Delta_k}-1).
\]
Thus,
\[
\|\hat{c}_k-c_k^{(0)}\|_2^2
=
(\rho_k^{(0)})^2
|s_ke^{i\Delta_k}-1|^2.
\]
Furthermore,
\[
|s_ke^{i\Delta_k}-1|^2
=
(s_k-1)^2
+
4s_k\sin^2(\Delta_k/2).
\]
Using the bounds on \(s_k\) and \(\Delta_k\), we obtain
\[
\|\hat{c}_k-c_k^{(0)}\|_2
\leq
\rho_k^{(0)}
\kappa(\alpha_{\theta},\alpha_{\rho}),
\]
where
\[
\kappa(\alpha_{\theta},\alpha_{\rho})
\coloneqq
\sqrt{
(e^{\alpha_{\rho}}-1)^2
+
4e^{\alpha_{\rho}}\sin^2(\alpha_{\theta}/2)
}.
\]
For small \(\alpha_{\theta}\) and \(\alpha_{\rho}\), we have the approximation
\[
\kappa(\alpha_{\theta},\alpha_{\rho})
\approx
\sqrt{\alpha_{\theta}^2+\alpha_{\rho}^2}.
\]
This shows that phase and radial corrections jointly induce a controlled local Euclidean perturbation.

For the entire \(U\)-subspace, let \(c^{(0)}\) be the concatenation of all two-dimensional blocks. Orthogonality across blocks gives
\[
\|\hat{c}-c^{(0)}\|_2
\leq
\kappa(\alpha_{\theta},\alpha_{\rho})
\|c^{(0)}\|_2.
\]
If the \(V\)-side correction is further controlled by norm clipping or regularization such that
\[
\|\hat{v}-v^{(0)}\|_2\leq \beta_v,
\]
then the overall unnormalized correction satisfies
\[
\|(\hat{c},\hat{v})-(c^{(0)},v^{(0)})\|_2
\leq
\kappa(\alpha_{\theta},\alpha_{\rho})\|c^{(0)}\|_2
+
\beta_v.
\]
Since \(\|c^{(0)}\|_2\leq 1\), we further have
\[
\|(\hat{c},\hat{v})-(c^{(0)},v^{(0)})\|_2
\leq
\kappa(\alpha_{\theta},\alpha_{\rho})
+
\beta_v.
\]
Define the effective perturbation radius
\[
\varepsilon_{\mathrm{eff}}
\coloneqq
\kappa(\alpha_{\theta},\alpha_{\rho})
+
\beta_v.
\]
When \(\varepsilon_{\mathrm{eff}}<1\), Lemma~\ref{lem:bounded-correction} applies directly to the Stage-II bounded residual refinement.

If no explicit \(V\)-side norm clipping is used in implementation, the quantity \(\|\hat{v}-v^{(0)}\|_2\) can instead be monitored empirically and controlled through \(\alpha_v\), regularization, or early stopping. In this case, Lemma~\ref{lem:bounded-correction} provides a conditional guarantee: as long as the realized correction norm remains small, the source-modality similarity structure is stably preserved.

\subsubsection{Implication}
\label{app:bounded-correction-implication}

The above analysis shows that the bounded parameterization of Stage II is not merely an engineering choice. It provides a tunable mechanism for balancing semantic preservation and target-distribution compatibility. Larger values of \(\alpha_{\theta}\), \(\alpha_{\rho}\), and \(\alpha_v\) allow stronger geometric correction, but may increase the risk of distorting the semantic structure of the source modality. Smaller correction scales better preserve the source geometry, but may be insufficient for entering the target-modality support. Effective alignment therefore requires a structured trade-off between correction strength and semantic stability.

\subsection{Geometric Motivation for the Periodic Phase Prior}
\label{app:phase-prior-motivation}

The previous sections explain why dominant residual directions should be corrected and why the correction should be bounded. We now explain why the proposed method uses two-dimensional blockwise polar coordinates in the dominant subspace and learns a target-modality prior in phase space.

\subsubsection{Two-Dimensional Blockwise Polar Decomposition}
\label{app:blockwise-polar}

Within the dominant subspace \(U\), the method partitions the projected coordinates into multiple two-dimensional blocks. For the \(k\)-th block, let its Cartesian coordinate be \((a_k,b_k)\). The corresponding polar coordinates are
\[
\rho_k
=
\sqrt{a_k^2+b_k^2+\epsilon},
\qquad
\theta_k
=
\mathrm{atan2}(b_k,a_k).
\]
Here, \(\rho_k\) represents the radial magnitude or energy of the block, whereas \(\theta_k\) represents its direction or phase.

This decomposition separates geometric variation in the dominant subspace into two components. Radial variation controls the amplitude or energy of each block, while phase variation controls the directional structure within each block. Since the phase variable lies on the periodic domain \([-\pi,\pi)\), it naturally has circular geometry. Consequently, directly modeling phase variables with ordinary Euclidean Gaussian noise is not fully appropriate; a more natural choice is to use wrapped Gaussian noise or other periodic distributions.

\subsubsection{Phase Marginals and Phase Couplings}
\label{app:phase-marginals-couplings}

The phase distribution of the target image modality contains two types of information. The first is the marginal phase preference of each two-dimensional block, which can be represented by the circular mean
\[
\bar{\psi}_k
=
\arg\!\left(
\mathbb{E}[e^{i\theta_k^{(x)}}]
\right).
\]
If \(\theta_k^{(x)}\) is close to uniformly distributed on the circle, then \(\left|\mathbb{E}[e^{i\theta_k^{(x)}}]\right|\) is close to \(0\). If it is concentrated around a certain direction, this magnitude becomes large.

The second type of information is the dependency between phase differences across different blocks. We define
\[
M_{k\ell}^{(x)}
=
\mathbb{E}\!\left[
e^{i(\theta_k^{(x)}-\theta_{\ell}^{(x)})}
\right].
\]
The magnitude \(|M_{k\ell}^{(x)}|\) measures the consistency of the phase difference between the \(k\)-th and \(\ell\)-th blocks, while \(\arg(M_{k\ell}^{(x)})\) gives the empirical phase offset.

Therefore, the marginal anchors \(\bar{\psi}_k\), block weights \(\alpha_k\), pairwise coupling strengths \(A_{k\ell}\), and phase offsets \(\eta_{k\ell}\) constructed in Stage I can be interpreted as low-order statistics of the target image modality in phase space.

\subsubsection{Periodic Potential and Drift Field}
\label{app:periodic-potential}

Based on the marginal and pairwise phase statistics above, we define the periodic potential
\[
\Psi(\phi)
=
\sum_k
\alpha_k
\left[
1-\cos(\phi_k-\bar{\psi}_k)
\right]
+
\sum_{(k,\ell)\in E}
A_{k\ell}
\left[
1-\cos(\phi_k-\phi_{\ell}-\eta_{k\ell})
\right].
\]
This potential contains two types of terms. The first encourages each phase variable to approach its target-modality marginal phase anchor. The second encourages phase differences between blocks to follow the dependency structure observed in the target modality.

Taking the gradient with respect to \(\phi_k\) gives the drift field used in the main text:
\[
[\nabla_{\phi}\Psi(\phi)]_k
=
\alpha_k\sin(\phi_k-\bar{\psi}_k)
+
\sum_{\ell:(k,\ell)\in E}
A_{k\ell}
\sin(\phi_k-\phi_{\ell}-\eta_{k\ell}).
\]
Thus, the periodic phase prior in Stage I can be interpreted as a local geometric constraint field in periodic phase space, constructed from the internal phase statistics of the target image modality.

\subsubsection{Wrapped Gaussian Score Prior}
\label{app:wrapped-gaussian-score}

Because phase variables are periodic, perturbed phases must be mapped back to \([-\pi,\pi)\) through a wrap operation. The wrapped Gaussian provides a natural local noise model on periodic domains.

Given a phase vector \(\phi\), we first construct the drifted phase center
\[
\mu_{\phi}
=
\mathrm{wrap}\!\left(
\phi-\tau\nabla_{\phi}\Psi(\phi)
\right).
\]
Then we sample a perturbed phase vector
\[
\tilde{\phi}
=
\mathrm{wrap}\!\left(
\mu_{\phi}
+
\sqrt{2}\sigma_t\epsilon
\right),
\qquad
\epsilon\sim \mathcal{N}(0,I).
\]
The phase score prior \(s_{\phi}\) is trained to predict the score of this wrapped Gaussian distribution in phase space. As a result, \(s_{\phi}\) does not directly learn a text-to-image mapping; instead, it learns the internal periodic phase structure of the target image modality.

\subsubsection{Implication}
\label{app:phase-prior-implication}

The Stage-I phase prior is not an arbitrary additional module. It follows from three observations: the residual modality gap is direction-dependent; the dominant subspace can be organized into two-dimensional blocks and expressed in polar coordinates; and the target modality exhibits estimable internal statistics over phase variables and phase differences. Therefore, the periodic phase prior provides a target-modality geometric constraint for the Stage-II bounded correction. It guides the source representation toward target-modality compatibility without relying on unconstrained distribution matching.

\section{Experiment Details} \label{app:2}

\subsection{Setting}

\ding{182}. \textbf{Geometric Level.} We follow the diagnostic setting in Sec.~\ref{sec:3}. The target modality is the image representation set \(X\), and the source modality is the text representation set \(Y\). Given paired image-text representations \(\{(x_i,y_i)\}_{i=1}^{n}\), where \(x_i\in X\), \(y_i\in Y\), and the two representations are semantically matched, all embeddings are evaluated in the shared normalized representation space and are \(\ell_2\)-normalized before metric computation. To avoid leaking pairwise correspondence into the alignment process, we separate the data into two parts. The first part is used as the statistic-estimation set, where only the marginal distributions of image and text representations are used to estimate the statistics required by each method, such as means, covariance-related quantities, subspaces, or residual structures. No image-text pairing information is used in this stage. The second part is used as a held-out paired diagnostic set. It is used only for evaluation. For any alignment method \(T\), we transform each source representation \(y_i\) into a target-modality substitute representation \(z_i=T(y_i)\), and evaluate the relation among the original source representation \(y_i\), the transformed representation \(z_i\), and the target representation \(x_i\) on the same held-out pairs. Unless otherwise specified, all metrics are computed on 10K held-out paired samples, and \(k=20\) is used for nearest-neighbor-based metrics.

\ding{183}. \textbf{MLLM Level.} We use Llama-3-8B-Instruct as the language-model backbone and connect modality features to the LLM through a two-layer MLP projector with GELU activation. In our setting, the aligned text representations are treated as substitute visual tokens. These text-induced representations are first projected by the MLP into the LLM embedding space and then used as visual-style inputs for multimodal training. The training procedure follows a two-stage pipeline. \ding{182} Modality Substitution Pretraining, we train only the projector on the filtered Bunny-1M dataset for 1 epoch, with the LLM frozen. The learning rate is set to \(5\times10^{-4}\). \ding{183} Visual Instruction-Tuning, we initialize the projector from the first stage and conduct full-parameter fine-tuning on InternVL-Chat-V1.2 for 1 epoch. The learning rate is reduced to \(1\times10^{-5}\). All experiments are performed on 8 NVIDIA H200 GPUs. With approximately 2.2M training samples in total, the full training pipeline takes about 12 hours.

\subsection{Metrics}

\textbf{Source Modality.} This evaluation does not rely on downstream tasks or human semantic labels. Instead, it directly measures self-consistency in the representation space. We use the original source representations \(Y\) as the semantic reference and compare them with the transformed representations \(Z=T(Y)\). \ding{182} First, we measure instance-level semantic consistency: \(\Phi(T)=\frac{1}{n}\sum_{i=1}^{n}y_i^\top z_i\). Since all representations are normalized, \(y_i^\top z_i\) is the cosine similarity between the original source representation and its transformed substitute. A larger \(\Phi(T)\) indicates that the transformed representation remains close to its original source semantic position. \ding{183} Second, we measure whether the relative geometry within the source modality is preserved. For a randomly sampled pair set \(\mathcal P\), we define \(\Psi(T)=\operatorname{corr}(\{y_i^\top y_j\}_{(i,j)\in\mathcal P},\{z_i^\top z_j\}_{(i,j)\in\mathcal P})\). A larger \(\Psi(T)\) means that pairwise semantic relations in the source modality remain stable after transformation. \ding{184} Third, we measure local neighborhood consistency. Let \(\mathcal N_k^Y(y_i)\) denote the \(k\)-nearest-neighbor set of \(y_i\) in the original source space \(Y\), and let \(\mathcal N_k^Z(z_i)\) denote the \(k\)-nearest-neighbor set of the same sample after transformation. We define \(\Omega_k(T)=\frac{1}{n}\sum_{i=1}^{n}|\mathcal N_k^Y(y_i)\cap \mathcal N_k^Z(z_i)|/k\). A larger \(\Omega_k(T)\) indicates that local semantic neighborhoods are better preserved after alignment.

\textbf{Target Modality.} Second, we evaluate local modality mixing, which captures whether \(X\) and \(Z\) are locally interleaved rather than only globally close. Let \(\mathcal Q=X\cup Z\). For any \(u\in\mathcal Q\), let \(\mathcal N_k(u)\) denote its \(k\)-nearest-neighbor set in \(\mathcal Q\setminus\{u\}\), and define the local target-modality proportion as \(p_k(u)=\frac{1}{k}\sum_{v\in\mathcal N_k(u)}\mathbf 1[v\in X]\). Here, \(\mathbf 1[v\in X]\) is not a semantic label, but a modality-origin indicator. We use the binary entropy \(H_2(p)=-p\log_2p-(1-p)\log_2(1-p)\) to measure local modality mixing, and normalize the score by the expected value under random permutation of modality-origin labels. We report two directional scores: \(M_k^Z(T)\) and \(M_k^X(T)\). \(M_k^Z(T)\) measures whether transformed source representations enter the support region of the target modality, while \(M_k^X(T)\) measures whether the target-modality support is covered by transformed source representations. For each method, we define \(r_i^T=x_i-z_i\), and compute the residual covariance \(\Sigma_r^T=\frac{1}{n}\sum_{i=1}^{n}r_i^T(r_i^T)^\top\). We examine the normalized residual eigenvalue spectrum \(\lambda_j(\Sigma_r^T)/\operatorname{tr}(\Sigma_r^T)\), and compute the residual anisotropy ratio \(A_r(T)=\lambda_{\max}(\Sigma_r^T)/(\operatorname{tr}(\Sigma_r^T)/d)\). A smaller \(A_r(T)\) and a less concentrated residual spectrum indicate that the method better suppresses the dominant anisotropic residual directions identified in Sec.~\ref{subsec:3.2}.

\subsection{Baselines}

\textbf{Unicorn.} Unicorn is a text-only data synthesis framework for VLM training. It constructs multimodal training data without real images through a three-stage pipeline: diverse caption synthesis, instruction-tuning data generation, and modality representation transfer. In particular, Unicorn first expands sparse caption seeds into diverse captions, then generates instruction-tuning data from these captions, and finally transfers text representations encoded by LLM2CLIP into the visual representation space to obtain synthetic image representations. In our experiments, we use Unicorn as a text-only synthetic visual representation baseline, following its modality representation transfer setting to construct pseudo-visual features from text.

\textbf{C$^3$ Align.} C$^3$ is a simple training-free modality-gap correction baseline built on the Connect-Collapse-Corrupt principle. The Connect step assumes that multimodal contrastive learning has already placed related concepts from different modalities into a shared representation space. The Collapse step removes the dominant modality gap by subtracting modality-specific embedding means. The Corrupt step injects Gaussian noise as regularization to improve robustness under alignment noise. In our setting, given a source text representation \(y\), we first shift it toward the target image centroid as \(y_\mu=y-\mu_y+\mu_x\), then add Gaussian perturbation and normalize the result to obtain the aligned substitute representation.

\textbf{ReAlign.} ReAlign is a training-free statistical alignment baseline that maps source-modality representations into the target-modality distribution using low-order statistics estimated from unpaired data. It consists of three closed-form steps: Anchor Alignment, Trace Alignment, and Centroid Alignment. First, Anchor Alignment removes the first-order mean bias by shifting the centered source representation to the target centroid. Second, Trace Alignment rescales the centered source residual using a global trace-matching factor, thereby matching the target residual energy while preserving the source covariance structure. Finally, after spherical projection, Centroid Alignment corrects the induced centroid drift and re-normalizes the representation on the unit sphere. In our experiments, we apply this operator to text representations to obtain ReAlign substitute visual representations.

\subsection{Evaluation Setting}

We evaluate the model on a broad set of multimodal benchmarks covering three aspects of visual understanding. \ding{182} For \textbf{General Perception}, we use MME \cite{fu2023mme} test, MMStar \cite{chen2024we}, ScienceQA \cite{lu2022learn}-image dev\&test, and RealWorldQA. For \ding{183} \textbf{Complex reasoning}, we evaluate on MMMU \cite{yue2024mmmu} validation single-image, MMMU-Pro\cite{yue2025mmmu} single-image, VisuLogic \cite{xu2025visulogic} train, and LogicVista \cite{xiao2024logicvista}. For \ding{184} \textbf{Hallucination assessment}, we use CRPE \cite{wang2023all}, POPE \cite{li2023evaluating}, and HallusionBench \cite{guan2024hallusionbench}. Across all benchmarks, we report accuracy (acc) as the unified evaluation metric, enabling a consistent comparison among different methods.




\section{Applicability}

Our analysis and method are built on the premise that the source and target modalities are embedded into a shared normalized representation space produced by a pretrained multimodal contrastive encoder. In this setting, the modality gap is assumed to arise within an already semantically compatible space: the two modalities share dominant geometric structure, while the remaining discrepancy appears as a structured anisotropic residual. This premise is important because AnisoAlign is designed to correct such residual geometric mismatch, rather than to align two arbitrary or unrelated distributions from scratch. Therefore, when the pretrained encoder fails to establish a meaningful shared semantic space, or when the source and target modalities do not exhibit compatible dominant geometry, the modality-gap structure studied in this work may become weak or absent, and explicit anisotropic correction may be less effective.

\end{document}